\documentclass[%
 reprint,
superscriptaddress,
 amsmath,amssymb,
 aps,
prl,
letterpaper
]{revtex4-2}
\usepackage[utf8]{inputenc}

\usepackage[caption=false]{subfig}
\usepackage{graphicx}
\usepackage{here}
\usepackage{epstopdf}
\usepackage{array}
\usepackage{physics}
\usepackage{verbatim}
\usepackage{amsmath,amsfonts,amssymb,amscd,mathtools}
\usepackage{amsthm}
\usepackage{tabularx}
\usepackage{stmaryrd}
\usepackage{enumerate}
\usepackage{wasysym}
\usepackage{braket}
\usepackage{mathrsfs}
\usepackage{microtype}
\usepackage{hyperref}
\usepackage{soul}
\usepackage[dvipsnames]{xcolor}
\hypersetup{
    bookmarksnumbered=true, 
    unicode=false, 
    pdfstartview={FitH}, 
    pdftitle={}, 
    pdfauthor={}, 
    pdfsubject={}, 
    pdfcreator={}, 
    pdfproducer={}, 
    pdfkeywords={}, 
    pdfnewwindow=true, 
    colorlinks=true, 
    linkcolor=NavyBlue, 
    citecolor=NavyBlue, 
    filecolor=NavyBlue, 
    urlcolor=NavyBlue 
}
\usepackage{tikz}
\usepackage[lmargin=.7in,rmargin=.7in,tmargin=.7in,bmargin=1in]{geometry}
\usepackage{newtxtext,newtxmath}

\usepackage{cleveref}

\usepackage{algorithm}
\usepackage{algorithmicx}
\usepackage{algpseudocode}



\usepackage{booktabs}

\theoremstyle{plain}
\newtheorem{thm}{Theorem}
\newtheorem{cor}[thm]{Corollary}
\newtheorem{lem}[thm]{Lemma}
\newtheorem{pro}[thm]{Proposition}

\theoremstyle{definition}
\newtheorem{defn}[thm]{Definition}

\usepackage[most,breakable]{tcolorbox}
\tcbset{ sharp corners}
  {\expandafter\ifstrequal\expandafter{#1}{filled}{\begin{tcolorbox}[colback=MidnightBlue!70!black!70!TealBlue!2!white,colframe=MidnightBlue!70!black!70!TealBlue!30!white,breakable,enhanced,left=5.75pt,right=5.75pt,grow sidewards by=10pt]}{\begin{tcolorbox}[colback=white,colframe=gray!0,breakable,enhanced,left=5.75pt,right=5.75pt,grow sidewards by=10pt]}}%
  {\end{tcolorbox}}

\newenvironment{defnboxed}[1][white]
  {\expandafter\ifstrequal\expandafter{#1}{filled}{\begin{tcolorbox}[colback=MidnightBlue!70!black!70!TealBlue!2!white,colframe=MidnightBlue!70!black!70!TealBlue!30!white,breakable,enhanced,left=5.75pt,right=5.75pt,grow sidewards by=10pt]}{\begin{tcolorbox}[colback=MidnightBlue!10,colframe=MidnightBlue!0,breakable,enhanced,left=5.75pt,right=5.75pt,grow sidewards by=10pt]}}%
  {\end{tcolorbox}}

\newenvironment{proboxed}[1][white]
  {\expandafter\ifstrequal\expandafter{#1}{filled}{\begin{tcolorbox}[colback=MidnightBlue!70!black!70!TealBlue!2!white,colframe=MidnightBlue!70!black!70!TealBlue!30!white,breakable,enhanced,left=5.75pt,right=5.75pt,grow sidewards by=10pt]}{\begin{tcolorbox}[colback=OliveGreen!10,colframe=OliveGreen!0,breakable,enhanced,left=5.75pt,right=5.75pt,grow sidewards by=10pt]}}%
  {\end{tcolorbox}}

  \newenvironment{thmboxed}[1][white]
  {\expandafter\ifstrequal\expandafter{#1}{filled}{\begin{tcolorbox}[colback=MidnightBlue!70!black!70!TealBlue!2!white,colframe=MidnightBlue!70!black!70!TealBlue!30!white,breakable,enhanced,left=5.75pt,right=5.75pt,grow sidewards by=10pt]}{\begin{tcolorbox}[colback=Maroon!10,colframe=Maroon!0,breakable,enhanced,left=5.75pt,right=5.75pt,grow sidewards by=10pt]}}%
  {\end{tcolorbox}}
  
   \newenvironment{lemboxed}[1][white]
  {\expandafter\ifstrequal\expandafter{#1}{filled}{\begin{tcolorbox}[colback=MidnightBlue!70!black!70!TealBlue!2!white,colframe=MidnightBlue!70!black!70!TealBlue!30!white,breakable,enhanced,left=5.75pt,right=5.75pt,grow sidewards by=10pt]}{\begin{tcolorbox}[colback=Gray!20,colframe=Maroon!0,breakable,enhanced,left=5.75pt,right=5.75pt,grow sidewards by=10pt]}}%
  {\end{tcolorbox}}

\newcommand{\eq}[1]{(\hyperref[eq:#1]{\ref*{eq:#1}})}

\renewcommand{\sec}[1]{\hyperref[sec:#1]{Section~\ref*{sec:#1}}}
\newcommand{\thrm}[1]{\hyperref[thrm:#1]{Theorem~\ref*{thrm:#1}}}
\newcommand{\lemm}[1]{\hyperref[lemm:#1]{Lemma~\ref*{lemm:#1}}}
\newcommand{\prop}[1]{\hyperref[prop:#1]{Proposition~\ref*{prop:#1}}}
\newcommand{\corr}[1]{\hyperref[corr:#1]{Corollary~\ref*{corr:#1}}}
\newcommand{\fig}[1]{\hyperref[fig:#1]{~\ref*{fig:#1}}}
\newcommand{\deff}[1]{\hyperref[deff:#1]{~\ref*{deff:#1}}}

\newcommand{\mA}{\mathcal{A}}
\newcommand{\mC}{\mathcal{C}}
\newcommand{\mE}{\mathcal{E}}
\newcommand{\mN}{\mathcal{N}}
\newcommand{\mU}{\mathcal{U}}

\newcommand{\mD}{\mathcal{D}}

\newcommand{\mL}{\mathcal{L}}

\newcommand{\mG}{\mathcal{G}}
\newcommand{\mH}{\mathcal{H}}
\newcommand{\mM}{\mathcal{M}}

\newcommand{\mB}{\mathcal{B}}

\newcommand{\mP}{\mathcal{P}}

\newcommand{\mX}{\mathcal{X}}
\newcommand{\mY}{\mathcal{Y}}
\newcommand{\mZ}{\mathcal{Z}}

\newcommand{\mbD}{\mathbb{D}}

\newcommand{\mbI}{\mathbb{I}}

\newcommand{\mbM}{\mathbb{M}}

\newcommand{\mbR}{\mathbb{R}}

\newcommand{\mfS}{\mathfrak{S}}

\newcommand{\aX}{{\abs{\mX}}}
\newcommand{\aU}{{\abs{\mU}}}

\newcommand{\sD}{\widetilde{D}}

\newcommand{\pI}{\overline{I}}
\newcommand{\sI}{\widetilde{I}}
\newcommand{\pD}{\overline{D}}

\newcommand{\sR}{\widetilde{R}}
\newcommand{\vlam}{{\Vec{\lambda}}}

\newcommand{\mW}{\mathcal{W}}

\newcommand{\ve}{\varepsilon}

\DeclareMathOperator{\supp}{supp}

\DeclareMathOperator{\spec}{spec}

\newcommand{\ba}{\begin{eqnarray}}
\newcommand{\ea}{\end{eqnarray}}
\newcommand{\bann}{\begin{eqnarray*}}
\newcommand{\eann}{\end{eqnarray*}}
\newcommand{\bal}{\begin{equation}\begin{aligned}}
\newcommand{\eal}{\end{aligned}\end{equation}}

\newcolumntype{L}[1]{>{\raggedright}p{#1}}
\newcolumntype{C}[1]{>{\centering}p{#1}}
\newcolumntype{R}[1]{>{\raggedleft}p{#1}}
\newcolumntype{D}{>{\centering\arraybackslash}X}

\usepackage{dcolumn}
\usepackage{bm}

\makeatletter
\renewcommand{\paragraph}[1]{\medskip\noindent\textbf{\textit{#1}}\ignorespaces}
\makeatother

\usepackage{xparse}
\newcommand{\rel}{\mathrm{Rel}}

\NewDocumentCommand{\strcon}{o}{\IfNoValueTF{#1}{\rel^\dagger}{\rel^{\dagger,#1}}}

\begin{document}

\title{Universality Sacrifices Reliability in Classical-Quantum Channel Coding}

\author{Kaito Watanabe}
\email{watanabe715@g.ecc.u-tokyo.ac.jp}
\affiliation{Department of Basic Science, The University of Tokyo, 3-8-1 Komaba, Meguro-ku, Tokyo 153-8902, Japan}
\affiliation{RIKEN Center for Quantum Computing (RQC), Hirosawa 2-1, Wako, Saitama 351-0198, Japan}

\author{Masahito Hayashi}
\email{hmasahito@cuhk.edu.cn}
\affiliation{School of Data Science, The Chinese University of Hong Kong, Shenzhen, Guangdong, 518172, China}
\affiliation{International Quantum Academy, Futian District, Shenzhen 518048, China}
\affiliation{Graduate School of Mathematics, Nagoya University, Chikusa-ku, Nagoya 464–8602, Japan}

\author{Takaya Matsuura}
\email{takayamatsuura@gmail.com}
\affiliation{RIKEN Center for Quantum Computing (RQC), Hirosawa 2-1, Wako, Saitama 351-0198, Japan}

\author{Hao-Chung Cheng}
\email{haochung@ntu.edu.tw}
\affiliation{Department of Electrical Engineering, National Taiwan University, Taipei 10617, Taiwan}
\affiliation{Physics/Mathematics Division, National Center for Theoretical Sciences, Taiwan}
\affiliation{Hon Hai (Foxconn) Quantum Computing Center, Taiwan}

\begin{abstract}
Universal channel coding enables communication without a complete description of the channel. For classical channels, universal codes can attain both capacity and the optimal high-rate reliability. We show that this compatibility fails for classical-quantum channels in general; that is, the optimal reliability in the channel-aware scenario is not always achievable with universal coding due to the ignorance of the unitary rotation of the output system. We exhibit a family of classical-quantum channels for which one cannot achieve the channel-aware optimal reliability by a fixed coding scheme. We further derive a converse bound on the reliability for unitary-invariant decoders, a natural assumption for the universal coding scheme, that can be strictly smaller than the optimal channel-aware error exponent. Conversely, we construct a channel-independent encoder-decoder pair and establish a universally achievable bound on the reliability that matches this converse bound in the high-rate regime, thereby characterizing the optimal universal reliability. Specifically, the channel-aware and universal exponents are governed by the Petz and sandwiched Rényi divergences, respectively. These divergences coincide for commuting outputs but differ for noncommuting ones, explaining why universality preserves optimal reliability classically but can reduce it quantumly. Our results showcase the fundamental reliability cost of performing the classical-quantum channel coding task universally.
\end{abstract}
\maketitle

\paragraph{Introduction.}\,---\,
Transmitting information over a noisy channel is a central task in information theory. Shannon’s channel coding theorem~\cite{Shannon_mathematical} establishes that communication with an asymptotically vanishing error probability is possible up to a threshold rate, called the channel capacity. For classical channels, this capacity is characterized by the mutual information between the channel input and output, maximized over the input probability distribution.

Beyond capacity, a finer characterization concerns how rapidly the decoding error decreases with the number of channel uses. At a fixed transmission rate, the optimal exponential decay is quantified by the \emph{reliability function}. Classical coding theory provides an achievable lower bound on this function, known as the random-coding bound~\cite{Gallager_simple_derivation}, and a converse upper bound, known as the sphere-packing bound~\cite{Shannon_lower_bound, Haroutunian1968}. These bounds coincide in the high-rate regime, determining the optimal reliability, whereas a gap generally remains at low rates.

A separate question concerns the channel knowledge required to implement a communication scheme. Channel-dependent codes assume that the channel is known and tailor their encoding and decoding operations to its description. In practice, however, such a description may not be available. \emph{Universal coding} addresses this limitation by constructing encoding and decoding operations independently of the channel description, so that a single scheme can operate over a family of possible channels~\cite{Hayashi_2009_universal_coding, Bjelakovic_2009_classical_capacities, Matsuura_2025_universal_resolvability, Hayashi_2011_universally_attainable, Hayashi_2022_universal_cq_superposition, Boche_2020_universal_superposition, Berta_2017_compound_entanglement_assisted, Datta_2010_universal_private, Boche_entanglement-assisted_2017, watanabe2026needuniversalcorrelationdetector, rizzo2026universalquantumcoding, takagi2026universaldistillationquantumentanglement}. This requirement raises the question of whether a lack of channel knowledge compromises the achievable communication performance.

For classical channels, universality can be achieved without sacrificing optimal reliability in the high-rate regime. Specifically, universal codes attain the same random-coding exponent as channel-dependent codes~\cite{Gallager_simple_derivation}. Since this exponent coincides with the sphere-packing bound at high rates, universal codes achieve the optimal reliability in this regime. Thus, classical channel coding admits schemes that simultaneously attain universality and the best possible reliability.

For classical-quantum channels, which map classical inputs to quantum states, both optimal reliability and universal communication are separately attainable. When the channel is known, the sphere-packing bound~\cite{Dalai_2013_lower_bound} and the random-coding bounds~\cite{Renes_2025_tight_lower, Li_reliability_2025, Cheng_operator_layer_cake} coincide in the high-rate regime, determining the optimal reliability achievable by channel-dependent codes. Independently, universal coding theorems establish that the channel capacity can be achieved using encoding and decoding operations constructed without complete knowledge of the channel~\cite{Hayashi_2009_universal_coding, Bjelakovic_2009_classical_capacities, Matsuura_2025_universal_resolvability}. Thus, channel-aware codes attain the optimal high-rate reliability, while universal codes achieve capacity. Whether a single coding scheme can attain both properties, however, does not follow from either result.

In this paper, we show that these two goals are, in general, incompatible for classical-quantum channels. Unlike in classical channel coding, universal codes can incur a strict loss in reliability relative to channel-aware codes.

We first give an explicit qubit example for which ignorance of a common unitary rotation of the channel output states prevents universal attainment of the channel-aware random-coding exponent for any fixed encoder-decoder pair.
Moreover, we derive a sphere-packing upper bound for such decoders and show that it can lie strictly below the optimal channel-aware reliability in the high-rate regime. Consequently, no single universal coding scheme can attain the channel-aware error exponent throughout the family. Noncommutativity of the output states can therefore incur a strict reliability penalty from the absence of channel knowledge, especially the ignorance of the orientation of the output quantum system.

We further determine the optimal reliability compatible with universality. We construct an encoder-decoder pair independently of the channel description and establish a universal random-coding bound. In the high-rate regime, this achievable bound matches our sphere-packing bound for unitary-invariant decoders. Together, these results characterize the optimal high-rate reliability of universal classical-quantum channel coding and identify its separation from the channel-aware optimum.

This separation is reflected in the quantum entropies governing the underlying reliability functions. The channel-aware random-coding and sphere-packing bounds are expressed in terms of the Petz Rényi divergence, whereas the established universal counterparts involve the sandwiched Rényi divergence. These divergences coincide for commuting states, consistent with the compatibility of universality and optimal reliability in classical channel coding. For noncommuting outputs, their distinction underlies a strict gap between the optimal channel-aware and universal error exponents. Operationally, this gap reflects the absence of a reference frame specifying the orientation of the quantum output states: requiring communication to be insensitive to this orientation can strictly reduce its reliability. Universal capacity achievability therefore does not guarantee universal attainability of the optimal error exponent, revealing a fundamental distinction between communication over classical and quantum channels.

\paragraph{Classical-quantum channel coding}\,---\,
Let $\mX$ be the finite set of classical symbols, and $\mD(\mH)$ be the set of density matrices on a finite-dimensional Hilbert space $\mH$. 
A classical-quantum channel $W:\mX\to \mD(\mH), ~x\mapsto \rho_x$ is defined as a map, which maps a symbol $x$ to a quantum state $\rho_x\in\mD(\mH)$, labeled by $x$. Classical-quantum channel coding is a task of sending classical information through the classical-quantum channel reliably by performing preprocessing before sending, called the encoding, and post-processing after receiving the signal, called the decoding.

More precisely, an $(n, R,\ve)$ code is composed of a tuple $(\mC_n, \qty{\Lambda_n(m)}_{m=1}^{M_n})$ of the codebook $\mC_n\coloneqq \qty{x^n(1),\ldots, x^n(M_n)}\subseteq \mX^n$ with the number $M_n\coloneqq 2^{nR}$ of messages and the decoder $\qty{\Lambda_n(m)}_{m=1}^{M_n}$, a positive operator-valued measure (POVM) on system $\mD(\mH^{\otimes n})$ satisfying the following: 
For any $m\in\qty{1,\ldots, M_n}$, the average probability of decoding error
\bal
P^n_e(P, R, W)\coloneqq 1-\frac{1}{M_n}\sum_{m=1}^{M_n}\Tr[\Lambda_n(m)\rho_{x^n(m)}]\leq \ve.
\eal
If there exists a sequence of $(n, R, \ve_n)$ codes for sufficiently large $n$ satisfying $\lim_{n\to\infty}\ve_n=0$, $R$ is referred to an achievable rate. The capacity $C(W)$ of $W$ is defined as the supremum of the achievable rate $R$.
Throughout this paper, we consider the constant-composition codebook $\mC_n$ whose codewords $\qty{x^(1),\ldots, x^n(M_n)}$ have the same type, i.e., the frequency of each symbol in $\mX$ is the same.

The error exponent of the classical-quantum channel coding offers a more refined analysis of \emph{reliability}, that is, how fast one can suppress the decoding error while achieving the target communication rate.
The reliability function, or the error exponent $E(P,R, W)$ of the classical-quantum channel coding, with a target communication rate $R$, is defined as the supremum of the limit $\limsup_{n\to\infty}\frac{-1}{n}\log \ve_n,$
for the sequence of $(n,R_n,\ve_n)$-codes with $\liminf_{n\to\infty}R_n\geq R$. 

In the classical case, the capacity and the reliability of the classical channel are almost established: The asymptotically optimal rate over the classical channel $W$ when the input distribution is fixed to $P$ is given as $I(P;W)$, where $P=P(x)$ is the fixed type for the codebook, and $I(P;W)\coloneqq H(\sum_xP(x)W(\cdot|x))-\sum_{x}P(x)H(W(\cdot|x))$ is the mutual information between the input and the output.
Here, $H(P)\coloneqq -\sum_x P(x)\log P(x)$ is the Shannon entropy.
The tight characterization of the reliability function is obtained only for the high-rate regime above a certain critical communication rate. 
The achievable bound for the reliability function is called the random-coding bound $E(P,R,W)\geq \sup_{1/2<\alpha<1}\frac{\alpha-1}{\alpha}(R-I^{\rm Aug}_\alpha(P; W)),$
where $I^{\rm Aug}_\alpha(P;W)\coloneqq \inf_Q\sum_xp(x)D_\alpha(W(\cdot | x)\|Q(\cdot))$ is the Augustin information~\cite{Gallager_simple_derivation}. Here, $D_\alpha(P(x)\|Q(x))\coloneqq \frac{1}{\alpha-1}\sum_xP(x)^\alpha Q(x)^{1-\alpha}$ is the classical Rényi $\alpha$ divergence.
On the other hand, the converse bound for the reliability function is the sphere-packing bound $E(P,R,W)\leq \sup_{0<\alpha<1}\frac{\alpha-1}{\alpha}(R-I^{\rm Aug}_\alpha(P; W))$~\cite{Shannon_lower_bound, Haroutunian1968}.
As can be seen from the expression, the matching characterization of the reliability function is not known yet: In the high-rate regime $R_{\rm crit}\leq R\leq I(P;W)$ for a critical rate $R_{\rm crit}$, the two bounds match, but in the low-rate regime $0<R<R_{\rm crit}$, the two bounds do not coincide. Characterizing the exact reliability for the low-rate regime is a longstanding open problem even in the classical case.

Similar to the classical case, the optimal communication rate of the classical-quantum channel is characterized by the mutual information $I(P;W)$ of the channel~\cite{Holevo_1996_capacity, SW_theorem}.
Here, the mutual information of the classical-quantum channel is defined as $I(P;W)\coloneqq S(\sum_xP(x)\rho_x)-\sum_{x}P(x)S(\rho_x)$, and $S(\rho)\coloneqq -\Tr[\rho\log \rho]$ is the von Neumann entropy.

Recent progress has established the corresponding bounds for the reliability of classical-quantum channel coding. 
The sphere packing bound for the classical-quantum channel coding is given as~\cite{Dalai_2013_lower_bound, Dalai_2017_constant_composition, Cheng_2019_sphere_packing}
\bal\label{eq: sphere-packing}
E(P,R,W)\leq \sup_{0< \alpha < 1}\frac{\alpha-1}{\alpha}(R-\pI^{\rm Aug}_\alpha(P:W)).
\eal
Here, $\pI^{\rm Aug}_\alpha(P; W)\coloneqq \inf_{\sigma}\sum_xP(x)\pD_\alpha(W_x\|\sigma)$ is a quantum extension of the Augustion information called the Petz Augustion information, and $\pD_\alpha(\rho\|\sigma)\coloneqq \frac{1}{\alpha-1}\log\Tr[\rho^\alpha\sigma^{1-\alpha}]$ is the Petz Rényi divergence \cite{Petz_1986_quasi_entropy}.

However, the random-coding bound for the classical-quantum channel cannot be easily derived because the output quantum states do not commute, and characterizing it was a longstanding open problem until recently.
Later, in Ref.~\cite{Cheng_operator_layer_cake}, it is shown that random-coding bound
\bal\label{eq: random coding bound}
E(P,R,W)\geq \sup_{1/2\leq\alpha<1}\frac{\alpha-1}{\alpha}(R-\pI^{\rm Aug}_\alpha(P;W)),
\eal
holds, which resolves the Burnashev-Holevo conjecture~\cite{Burnashev_1998reliability}.
Just as in the classical case, the bounds are tight in the high-rate regime, but not in the low-rate regime.
These results reveal the similarity in the optimal performance of coding for the classical channels and the classical-quantum channels.
Furthermore, similarly to the classical case, the sphere-packing bound and the random-coding bound match in the high-rate regime above a certain critical rate $R_{\rm crit}\leq R \leq I(P;W)$.

\paragraph{Universal coding}\,---\,
The coding schemes, encoder, and decoder that achieve optimal performance are generally tailored to the given channel. However, they may not be accessible to the sender and receiver, which motivates the search for \emph{universal coding}, designed independently of channel information yet achieving optimal performance.

Surprisingly, in the classical case, it is known that a coding scheme that devises the mutual information decoder achieves the channel capacity and the random coding bound universally~\cite{Csiszar_Korner_2011}. 
In other words, universal coding does not degrade the capacity and the reliability (in the high-rate regime) of the communication through an unknown classical channel at all.
From this, it seems natural to expect that one can also achieve optimal capacity and reliability universally for unknown classical-quantum channels.

So far, it is only known that one can achieve the capacity for the classical-quantum channel universally: In Ref.~\cite{Hayashi_2009_universal_coding, Bjelakovic_2009_classical_capacities}, it is shown that universal coding exists which achieves the channel capacity.
However, unlike in the classical case, it is not known whether one can universally achieve the random coding bound for the classical-quantum channels.

\paragraph{Random coding bound is not universally achievable}\,---\,
Our main interest here is whether we can achieve the random-coding bound in Eq.~\eqref{eq: random coding bound} with a fixed pair of the codebook and the decoder.
The best one can hope is that there exists such a pair which can achieve the random-coding bound for any classical-quantum channel, but remarkably, we show that it is not possible.

Fix a classical-quantum qubit channel $W:\qty{\pm}\to \ketbra{\psi_\pm}{\psi_{\pm}}$ where $\ket{\psi_{\pm}}\coloneqq\sqrt{p}\ket{0}\pm\sqrt{1-p}\ket{1}.$ 
We also consider a rotated channel $W_\theta:\qty{\pm}\to U_\theta\ketbra{\psi_\pm}{\psi_{\pm}}U^\dagger_\theta$, where $U_{\theta}\coloneqq{\rm diag}(1,e^{i\theta})$ is the phase-rotation.
If there exists a universal coding scheme which can achieve the random-coding bound for any classical-quantum channels, it should achieve the random-coding bound for all channels in the family $\qty{W_\theta}_{\theta\in [0,2\pi)}$ of rotated channels simultaneously.

It turns out that, even for this simple family of classical-quantum channels, one can never construct the universal coding scheme achieving the random-coding bound. In fact, taking the average of the decoding error over the phase of the rotation gives the converse bound, which is strictly less than the random-coding bound. 
Furthermore, for any target communication rate $R$, one can take a fixed phase rotation $\theta^*$ such that the decoding error never achieves the random-coding bound without knowledge of the phase rotation.
Further technical detail is deferred to End Matters.

This example shows that, in the channel-aware case, prior information about the orientation of the quantum states is actually crucial to achieve the random-coding bound in Eq.~\eqref{eq: random coding bound}, and impossible without it.

\paragraph{Unitary-invariant sphere-packing bound}\,---\,
The degradation of the reliability due to this ignorance can also be captured by considering the sphere-packing bound under the unitary invariance of the decoder, i.e., for any number $n$ of channel uses and for any message $m\in\qty{1,\ldots, M_n}$, the decoder $\qty{\Lambda_n(m)}_m$ satisfies
\bal
\Lambda_n(m)=(U^\dagger)^{\otimes n}\Lambda_n(m) U^{\otimes n}
\eal
The unitary invariance can be seen as a natural requirement for the universal codes, since the performance of the universal code should not vary under an arbitrary unitary rotation $W\mapsto U W U^\dagger$ of the output.
In fact, the universal decoders proposed in the previous works~\cite{Hayashi_2009_universal_coding, Matsuura_2025_universal_resolvability} also satisfy the unitary invariance.

Remarkably, it turns out that the symmetry constraint imposed on the decoder crucially degrades the reliability as follows.
\begin{thm}[Unitary-invariant sphere-packing bound for constant-composition codebooks]\label{thm: unitary invariant sphere-packing}
    Let $P_n$ be the $n$-type from which the codewords are drawn, converging to a fixed probability distribution $P$, and $W:\mX\to\mD(\mH)$ be a classical-quantum channel. Then, for any coding scheme with a unitary-invariant decoder, it holds that 
    \bal\label{eq: unitary invariant sphere-packing}
   \limsup_{n\to\infty}-\frac{1}{n}\log P^n_e(P_n, R, W)\leq \!\!\!\sup_{1/2\leq \alpha< 1} \!\! \frac{\alpha-1}{\alpha}\qty(R-\sI^{\rm Aug}_\alpha(P;W))
    \eal
    for the high-rate regime above a certain critical rate $R_{\rm crit}\leq R$.
    Here, $\sI^{\rm Aug}_\alpha(P:W)$ is the sandwiched Augustion information.
\end{thm}
The proof is given in Appendix~\ref{sup sec: sphere-packing bound}.
At rates above a certain critical rate, the sphere-packing bound under unitary invariance in Eq.~\eqref{eq: unitary invariant sphere-packing} is no larger than the channel-aware random-coding bound in Eq.~\eqref{eq: random coding bound}, and the inequality can be strict. This gap demonstrates a fundamental reliability cost of universal coding arising from the lack of knowledge of the freedom of the unitary rotation on the output system.

Since the sandwiched Rényi divergence is no larger than the Petz Rényi divergence, the left-hand side of Eq.~\eqref{eq: unitary invariant sphere-packing} is no larger than that of Eq.~\eqref{eq: sphere-packing}. This comparison also reflects the restriction imposed by unitary invariance.

Remarkably, the reduction of the sphere-packing bound due to the unitary invariance is clearly represented as a different type of quantum extension of Rényi divergence appearing in the form.
This phenomenon, where we obtain the sandwiched Rényi divergence instead of the Petz Rényi divergence due to the symmetric restriction, is also observed in other tasks such as work extraction under time-translation covariance~\cite{Watanabe_2026_Reliability}, and quantum hypothesis testing, a state-discrimination task, for the composite hypothesis with a certain unitary orbit~\cite{Hayashi_2026_Operational_interpretation}.

Let us briefly see the proof idea. In Ref.~\cite{Dalai_2013_lower_bound, Dalai_2017_constant_composition}, they derive the sphere-packing bound for the classical-quantum channel by reducing the error exponent of the channel coding to that of the hypothesis testing~\cite{hiai_1991_proper, Ogawa_2000_strong,nagaoka_2006_converse_theorem, Hayashi_2007_error_exponent, audenaert_2008, Mosonyi_Ogawa_2014}, a state-discrimination task, between the output of the channel and a dummy quantum state.
In our case where we consider the unitary-invariant decoder, the converse bound can be related to the hypothesis testing between the unitary-twirled quantum state and a dummy state.
Employing the results on the exponents of quantum hypothesis testing with a group symmetry and correlation~\cite{hiai_error_correlated_2008, Hiai_hypothesis_testing_group_2009, Lipka-Bartosil_quantum_dichotomies}, we establish Theorem~\ref{thm: unitary invariant sphere-packing}.

\paragraph{Universally achievable random coding bound}\,---\,
A natural question that arises now is the following: what is the universally achievable reliability? The known universally achievable reliability obtained in the previous papers~\cite{Hayashi_2009_universal_coding, Matsuura_2025_universal_resolvability} is not tight, as it does not give us the random-coding bound in the classical case, even though the random coding bound is universally achievable for the classical channel.

The following theorem exhibits a universally achievable bound for the error exponent of classical-quantum channel coding, which matches the sphere-packing bound with the unitary-invariant sphere-packing bound in \Cref{thm: unitary invariant sphere-packing}.

\begin{thm}[Universally achievable random-coding bound with constant composition codebooks]\label{thm: universal random coding bound}
    Let $P$ be a probability distribution, and $R$ be such that $0\leq R\leq H(P)$. There exists a fixed sequence of the codebooks with constant composition $\qty{P_n}_n$ satisfying $P_n\to P$, and the decoders such that the decoding error satisfies 
    \bal
    \liminf_{n\to\infty}-\frac{1}{n}\log P^n_e(P_n, R, W)\geq \!\! \sup_{1/2\leq \alpha< 1} \!\! \frac{\alpha-1}{\alpha}\qty(R-\sI^{\rm Aug}_\alpha(P:W))
    \eal
    for any classical-quantum channel $W$.
\end{thm}
For the proof, see Appendix~\ref{sup sec: random coding}.
We remark that \Cref{thm: universal random coding bound} is compatible with the previously established results about universal coding. 
That is, the right-hand side is positive whenever $R$ is below the mutual information $I(P;W)$, meaning that all the rates below the mutual information of the channel are universally achievable.
Furthermore, combining with the sphere-packing bound for the high-rate regime in Theorem~\ref{thm: unitary invariant sphere-packing}, the universal random coding bound with the sandwiched Rényi Augustin information gives us the tight characterization of the reliability function of universal classical-quantum channel coding.

 Remarkably, this result gives us a stark separation between the universal coding for the classical channel and the classical-quantum channel coding. In the classical-quantum case, ignorance about the channel degrades the reliability, which is represented as the difference in the emerging quantum Rényi divergence, while the universal coding scheme does not degrade the reliability in the classical case~\cite{Csiszar_Korner_2011}. Since the Petz--Augustin information and the sandwiched Augustin information coincide in the classical case, our results are reduced to the original observation for the classical case.

The central challenge in universal coding is to construct a decoder that works without knowledge of the channel. For example, the channel-aware decoder achieving the Petz random-coding bound in Ref.~\cite{Cheng_operator_layer_cake} requires both the optimizing Rényi parameter $\alpha$ and the spectral decompositions of $\rho_x^\alpha-v(\sum_x\rho_x)^\alpha$ for $0\leq v\leq 1$. These depend on the channel, making this construction difficult to adapt to the universal setting.

Our approach starts from a variational representation of the measured Rényi Augustin information as an optimization over families of positive operators $\qty{Y_x}_{x\in\mX}$. An optimal family defines a quotient decoder~\cite{Beigi_2024_lower_bound} that achieves the sandwiched random-coding bound. This reformulation expresses the decoder's channel dependence through the choice of the operator family, but does not by itself yield a universal decoder: the optimal family generally still depends on the channel.

To remove this dependence, we utilize Schur--Weyl duality to construct a channel-independent family of operators whose associated quotient decoder achieves the sandwiched random-coding bound up to a polynomial overhead. Since this overhead does not affect the asymptotic error exponent, the resulting decoder achieves the sandwiched random-coding bound universally.

\paragraph{Discussion}\,---\,
In this work, we clarified the separation between channel-aware and universal channel coding for the classical-quantum channel, which also exhibits a stark difference between the mathematical structures of classical and quantum channels.
Specifically, we present an example of a classical-quantum channel whose channel-aware random-coding bound is never achievable with any fixed codebook-decoder pair.
Furthermore, we identify the sphere-packing bound under unitary invariance of the decoder, a natural assumption for a universal decoder, as the sandwiched version of the sphere-packing bound. Moreover, we show that the sandwiched version of the random-coding bound is universally achievable, characterizing the optimal reliability of universal classical-quantum channel coding in the high-rate regime.

Our results show that, although channel capacity can be achieved without prior knowledge of the channel, such knowledge—in particular, knowledge of the orientation of the output quantum system with respect to the unitary rotation—is crucial for achieving optimal reliability. The symmetry constraint arising from the lack of this knowledge limits the achievable error exponent. 
Also, the new techniques developed here may also help characterize the ultimate performance of other quantum information tasks when the underlying states or channels are not fully known.

\paragraph{Disclosure of AI use.---} 
The authors used OpenAI’s ChatGPT models 5.6 Sol and 6 Astra to assist with exploring proof strategies, checking proofs, and preparing the manuscript. The authors reviewed and verified the outputs and take full responsibility for the content of this work.

\paragraph{Acknowledgments.---} 
The authors are indebted to Ryuji Takagi and Bartosz Regula for fruitful discussions.
We acknowledge the support of the JSPS KAKENHI Grant No. 26KJ0965, JST, PRESTO Grant Number JPMJPR24FA, and the World-Leading Innovative Graduate Study Program for Advanced Basic Science Course (WINGS-ABC) at the University of Tokyo.
HC acknowledges support from National Science and Technology Council (NSTC 115-2628-E-002-005, NSTC 114-2119-M-001-002, and NSTC 115-2124-M-002-014) and Ministry of Education (NTU-115V2016-1, NTU-CC115L893705, and NTU-115L900702).

\let\oldaddcontentsline\addcontentsline
\renewcommand{\addcontentsline}[3]{}

\bibliographystyle{apsrmp4-2}
\bibliography{myref}

\let\addcontentsline\oldaddcontentsline


\newpage
\newgeometry{hmargin=1.2in,vmargin=0.8in}

\widetext
\begin{center}
{\large \bf End Matters}
\end{center}

In the following, we show that the channel-aware random-coding bound with Petz Rényi Augustin information is never universally achievable by exhibiting a concrete example of a family of channels.

Fix $0<p<1/2$ and consider the binary qubit channel
\begin{equation}
 W:\qty{\pm}\mapsto\ketbra{\psi_\pm},\qquad
 \ket{\psi_\pm}=\sqrt{1-p}\,\ket0\pm\sqrt p\,\ket1.
 \label{eq:channel}
\end{equation}
The actual channel is $W_\theta:\qty{\pm}\mapsto U_\theta\ketbra{\psi_\pm}{\psi_\pm} U_\theta^\dagger$, where
$U_\theta=\operatorname{diag}(1,e^{i\theta})$ and the same unknown
$\theta\in[0,2\pi]$ is used at every channel use.
The outputs do not commute, and the capacity is $h(p)>0$~\cite{Burnashev_1998reliability},
where $h(t)=-t\log t-(1-t)\log(1-t)$. All logarithms are natural.

A length-$n$ code consists of $M_n$ distinct words $x^n(m)\in\{+,-\}^n$
and positive decoding operators $\Lambda_{m}$ with
$\sum_m\Lambda_{m}=I$. Both the words and the operators are independent
of $\theta$. For equally likely messages, the average error probability is
\begin{equation}
 P^n_e(P,R,W_\theta)\coloneqq 1-\frac1{M_n}\sum_{m=1}^{M_n}
\Tr\qty[\Lambda_{m}
U^{\otimes n}_\theta\rho_{x^n(m)}(U^\dagger_\theta)^{\otimes n}].
 \label{eq:error}
\end{equation}
A sequence has rate $R$ when $n^{-1}\log M_n\to R$ and lower error
exponent $\liminf_n[-n^{-1}\log P_e(P,R,W_\theta)]$, with $-\log0=+\infty$.
Let $E_{\mathrm{aware}}(R)$ be the supremum of this over
rate-$R$ sequences allowed to depend on $\theta$; it is independent of
$\theta$, since a known rotation can be undone.
For the channels exhibied above, the channel-aware random coding bound is given by~\cite{Burnashev_1998reliability, Cheng_operator_layer_cake} 
\bal\label{eq:known}
E_{\rm aware}(P, R, W_\theta)= \sup_{0\leq s\leq 1}-\log\!\left[(1-p)^{1+s}+p^{1+s}\right]-sR.
\eal

Write $D(t\Vert p)=t\log(t/p)+(1-t)\log[(1-t)/(1-p)]$.

\begin{thmboxed}
   \begin{thm}[No universal attainment of the optimal channel-aware reliability]
\label{thm:universal-no-go}
For every $0<R<h(p)$, let $t\in(0,p)$ satisfy $h(t)=R$.
Every rate-$R$ code sequence independent of $\theta$ obeys
\begin{equation}
 \liminf_{n\to\infty}-\frac1n\log\max_{\theta\in[0,2\pi]}
 P^n_e(P,R,W_\theta)
 \;\le\;D(t\Vert p)\;<\;E_{\mathrm{aware}}(R).
 \label{eq:theorem}
\end{equation}
\end{thm} 
\end{thmboxed}
The maximum in Eq.~\eqref{eq:theorem} requires a reliability guarantee uniform
over the unknown phase; the phase is constant throughout each block.

\begin{proof}
Let $\Pi_k$ project onto binary basis vectors with $k$ ones, a space of
dimension $\binom nk$. Averaging over $\theta$ removes matrix elements
between different such spaces. Each codeword's block has trace
$w_k=\binom nk p^k(1-p)^{n-k}$ and is bounded by $w_k\Pi_k$ with respect to the Loewener bound.
Since the decoding operators sum to $I$, this block contributes at most
$w_k\min\{1,\binom nk/M_n\}$ to the average success probability. Hence
\begin{equation}
 \int_0^{2\pi}  P^n_e(P,R,W_\theta)\,\frac{d\theta}{2\pi}
 \ge\sum_{k=0}^n\binom nk p^k(1-p)^{n-k}
       \left(1-\frac{\binom nk}{M_n}\right)_+,
 \label{eq:average}
\end{equation}
where $(a)_+=\max\{a,0\}$.
For $0<q<t$, retain $k=\lfloor nq\rfloor$ in Eq.~\eqref{eq:average}.
Then $\binom nk=e^{nh(q)+o(n)}=o(M_n)$, and the binomial probability
is $e^{-nD(q\Vert p)+o(n)}$.
Since the maximum error is at least its phase average, taking the limit $n\to\infty$ and  $q\to t-0$
gives the first inequality in Eq.~\eqref{eq:theorem}.

Set $s=1-\log[(1-p)/p]/\log[(1-t)/t]\in(0,1)$. Substitution gives
\begin{align}
 D(t\Vert p)
 &=-(1-s)\log\!\left[(1-p)^{1/(1-s)}+p^{1/(1-s)}\right]-sR\notag\\
 &<-\log\!\left[(1-p)^{1+s}+p^{1+s}\right]-sR.
 \label{eq:strict}
\end{align}
Indeed, $u\mapsto\log[(1-p)^{1+u}+p^{1+u}]$ is strictly convex and
vanishes at zero; its slopes from zero at $s$ and $s/(1-s)>s$ give
the strict inequality. Since Eq.~\eqref{eq:known} is achievable,
Eq.~\eqref{eq:strict} proves $D(t\Vert p)<E_{\mathrm{aware}}(R)$.
\end{proof}

Theorem~\ref{thm:universal-no-go} already establishes the uniform no-go
statement. The following lemma yields the stronger fixed-phase conclusion
in Corollary~\ref{cor:fixed}.
\begin{lemboxed}
  \begin{lem}[A fixed phase]
\label{lem:fixed-rotation}
For every sequence of codes independent of $\theta$ for the phase-rotated
channel Eq.~\eqref{eq:channel}, with error defined by Eq.~\eqref{eq:error},
\begin{equation}
 \inf_{\theta\in[0,2\pi]}\liminf_n-\frac1n\log   P^n_e(P,R,W_\theta)
 =\liminf_n-\frac1n\log\max_{\theta\in[0,2\pi]}  P^n_e(P,R,W_\theta).
 \label{eq:fixed-rotation}
\end{equation}
\end{lem}  
\end{lemboxed}

\begin{proof}
The right-hand side of Eq.~\eqref{eq:fixed-rotation} is at most the left.
For the reverse inequality, let $F_n(z)$ be the concatenation
of the vectors
$M_n^{-1/2}(I-\Lambda_{m})^{1/2}\qty(\sum_{k=0}^nz^k\Pi_k)
\bigotimes_{i=1}^n\ket{\psi_{x_i(m)}}$.
This is a vector-valued polynomial of degree at most $n$ and,
by Eq.~\eqref{eq:error}, $P_e(P,R,W_\theta)=\|F_n(e^{i\theta})\|^2\le1$.
If the left-hand side of Eq.~\eqref{eq:fixed-rotation} is zero, equality is
immediate. Otherwise, choose any finite $\gamma>0$ below that side
and $0<\epsilon<\gamma$.

Here, let us define a set 
\bal
A_n\coloneqq\qty{\theta\in [0,2\pi)~|~P^n_e(P,R,W_\theta)\le e^{-n(\gamma-\varepsilon)}}.
\eal
Due to \Cref{thm:universal-no-go}, it holds that $A_n\subset A_{n+1}$ for any $n$, and any $\theta\in[0,2\pi]$ is included in the set $A_n$ for sufficiently large $n$.
For $|z|=r<1$, the mean-value inequality for the logarithm of a polynomial,
applied to each scalar projection of $F_n$, gives
\bal\label{eq:inside}
\log\|F_n(z)\|^2
 &\le\int_0^{2\pi}\frac{1-r^2}{|e^{i\theta}-z|^2}
                \log  P_e(P,R,W_\theta)\,\frac{d\theta}{2\pi}\\
                &\leq \frac{1+r}{1-r}\int_{A_n} \log  P_e(P,R,W_\theta)\,\frac{d\theta}{2\pi}+\frac{1+r}{1-r}\int_{A^c_n} \log  P_e(P,R,W_\theta)\,\frac{d\theta}{2\pi}\\
 &\le\frac{1-r}{1+r}\qty(-n(\gamma-\varepsilon))(1-o(1)).
\eal
Here, in the first inequality, we used the fact that $\log\|F_n(z)\|^2$ is a subharmonic function.
The second inequality is because it holds that 
\bal
\frac{1-r^2}{|e^{i\theta}-z|^2}\leq \frac{1+r}{1-r}.
\eal
The last inequality is because $\abs{A_n^c}\to 0$.
Due to Cauchy's coefficient estimate, the norm of the $k$th coefficient can be upper-bounded by $r^{-k}\max_{|z|=r}\|F_n(z)\|$. Summing up $n+1$
coefficients therefore yields
\begin{equation}
 \max_\theta  P_e(P,R,W_\theta)
 \le(n+1)^2r^{-2n}\max_{|z|=r}\|F_n(z)\|^2.
 \label{eq:outside}
\end{equation}
Thus the right-hand side of Eq.~\eqref{eq:fixed-rotation} is at least
$\gamma-\epsilon+2\log r$.
Let $r\to 1-0$ and $\epsilon\to +0$, then let $\gamma$ increase
to the left-hand side of Eq.~\eqref{eq:fixed-rotation}. This includes the case
where that side is infinite and proves equality.
\end{proof}

\begin{proboxed}
   \begin{cor}[No attainment on every fixed phase]
\label{cor:fixed}
Under the assumptions of Theorem~\ref{thm:universal-no-go},
\begin{equation}
 \exists\theta_*\in[0,2\pi]:\quad
 \liminf_{n\to\infty}-\frac1n\log  P_e(P,R,W_{\theta_*})
 <E_{\mathrm{aware}}(R).
\end{equation}
The phase $\theta_*$ may depend on the code sequence and $R$, but not on $n$.
\end{cor} 
\end{proboxed}

\begin{proof}
Apply Lemma~\ref{lem:fixed-rotation} to Eq.~\eqref{eq:theorem} and use its
strict gap. Attainment of the infimum over $\theta$ is unnecessary.
\end{proof}
\newpage
\appendix

\setcounter{thm}{0}
\renewcommand{\thethm}{S.\arabic{thm}}
\setcounter{figure}{0}
\renewcommand{\thefigure}{S.\arabic{figure}}

\begin{center}
{\large \bf Appendices}
\end{center}

\setcounter{tocdepth}{2}
\tableofcontents

\section{Preliminaries}

\subsection{Entropic quantities}
In the following, we showcase the entropic quantities relevant to the subsequent discussion.
The Shannon entropy $H(P)$ of the distribution $P$ is defined as
\bal
H(P)\coloneqq -\sum_{x\in\mX}P(x)\log P(x).
\eal
The Kullback-Leibler divergence of a probability distribution $P$ with respect to the probability distribution $Q$ is defined as
\bal
D(P\|Q)\coloneqq \sum_{x\in\mX}P(x)\log\frac{P(x)}{Q(x)}.
\eal
The classical Rényi relative entropy of a probability distribution $P$ with respect to another probability distribution $Q$ is defined as
\bal
D_\alpha(P\|Q)\coloneqq \frac{1}{\alpha-1}\log \sum_{x\in\mX}P(x)^\alpha Q(x)^{1-\alpha}.
\eal
The Umegaki relative entropy of a quantum state $\rho\in\mD(\mH)$ with respect to $\sigma\geq 0$ is defined as~\cite{Umegaki_relative}
\begin{equation}
    D(\rho\|\sigma)\coloneqq 
    \left\{\,
\begin{aligned}
    &\Tr[\rho\log \rho-\rho\log \sigma]~~~(\supp\rho\subset\supp\sigma)\\
    &+\infty~~~(\mbox{otherwise})
\end{aligned}
\right.
\end{equation}
The Petz Rényi relative entropy of $\rho\in\mD(\mH)$ with respect to $\sigma\geq 0$ of order $\alpha\in[0,1)\cp(1,\infty]$ is defined as~\cite{Petz_1986_quasi_entropy}
\begin{equation}
    \pD_\alpha(\rho\|\sigma)\coloneqq \left\{\,
    \begin{aligned}
        &\frac{1}{\alpha-1}\log\Tr\qty[\rho^\alpha\sigma^{1-\alpha}]~~~(\supp\rho\subset\supp\sigma\mbox{ or }\supp\rho\not\perp\supp\sigma, ~0<\alpha<1)\\
        &+\infty~~~(\mbox{otherwise})
    \end{aligned}
    \right.
\end{equation}

Another quantum extension of Rényi relative entropy, the sandwiched Rényi relative entropy of $\rho\in\mD(\mH)$ with respect to $\sigma\in\mD(\mH)$ of order $\alpha\in[0,1)\cup(1,\infty]$, is defined as~\cite{Muller_Lennert_2013_on_quantum_renyi, Wilde_2014_strong_converse}
\begin{equation}
    \sD_\alpha(\rho\|\sigma)\coloneqq \left\{\,
    \begin{aligned}
        &\frac{1}{\alpha-1}\log\Tr\qty[\qty(\sigma^{\frac{1-\alpha}{2\alpha}}\rho\sigma^{\frac{1-\alpha}{2\alpha}})^\alpha]~~~(\supp\rho\subset\supp\sigma\mbox{ or }\supp\rho\not\perp\supp\sigma, ~0<\alpha<1)\\
        &+\infty~~~(\mbox{otherwise})
    \end{aligned}
    \right.
\end{equation}
Note that both the Petz Rényi relative entropy and the sandwiched Rényi relative entropy coincide with the Umegaki relative entropy in the limit $\alpha\to 1$.

The measured Rényi relative entropy $D_\alpha^\mbM(\rho\|\sigma)$ of a quantum state $\rho$ with respect to the positive semidefinite matrix $\sigma\geq 0$ of order $\alpha\in[0,1)\cup(1,\infty]$ is defined as 
\bal
D_\alpha^\mbM(\rho\|\sigma)\coloneqq \sup_{\mM}D_\alpha(\mM(\rho)\|\mM(\sigma)).
\eal
Here, the supremum is taken over all the measurement channels on $\mD(\mH)$. 
The quasi-measured Rényi relative entropy $Q_\alpha^\mbM(\rho\|\sigma)$ is defined as $Q_\alpha^\mbM(\rho\|\sigma)\coloneqq \exp((\alpha-1)D_\alpha^\mbM(\rho\|\sigma))$.
In Ref.~\cite{Berta_2017_on_variational}, it is shown that the measured Rényi relative entropy of $\alpha\in (0,1)$ admits the following variational form:
\bal
 Q_\alpha^{\mbM}(\rho\|\sigma)&=\inf_{Y> 0}\qty(\Tr[\rho Y^{-\frac{\alpha-1}{\alpha}}])^\alpha\bigl(\Tr[\sigma Y]\bigr)^{1-\alpha}\\
 &=\inf_{Y> 0}
 \left\{\alpha\Tr\qty[\rho Y^{-\frac{\alpha-1}{\alpha}}]
 +(1-\alpha)\Tr[\sigma Y]\right\}.
 \label{eq:variational measured renyi}
\eal

The mutual information, a central quantity throughout this work, is defined through the relative entropies raised above.
For a probability distribution $P_X$ and a conditional probability distribution $W_{Y|X}(y|x)$, the classical mutual information $I(P;W)$ is defined as 
\bal
I(P;W)\coloneqq D((W\times P)_{XY}\|P_X(WP)_Y),
\eal
with $(W\times P)_{XY}(x,y)\coloneqq W_{Y|X}(y|x)P_X(x)$ and $(WP)_Y(y)\coloneqq \sum_{x\in\mX}W_{Y|X}(y|x)P_X(x)$.
In the classical-quantum case, the mutual information of the input distribution $P_X=P_X(x)$ and the classical-quantum channel $W:x\mapsto \rho_{B,x}$ is defined using the quantum state $\rho_{XB}=\sum_xP_X(x)\ketbra{x}{x}_X\otimes \rho_{B,x}$ as
\bal
I(P;W)\coloneqq D(\rho_{XB}\|\rho_X\otimes \rho_B).
\eal

Classical $\alpha$-Rényi Augustin information is defined as
\bal
I^{\rm Aug}_\alpha(P;W)\coloneqq \inf_{Q_Y\in\mP(\mY)}\sum_{x\in\mX}P_X(x)D_\alpha\qty(W_{Y|X}(\cdot|x)\| Q_Y(\cdot)).
\eal
For a quantum divergence $\mbD\in\qty{\pD,\sD,D^{\mbM}}$, the corresponding Rényi Augustin information is defined as 
\bal
\mbI^{\rm Aug}_\alpha(P;W)\coloneqq \inf_{\sigma\in\mD(\mH)}\sum_{x\in\mX}P_X(x)\mbD_\alpha\qty(\rho_x\| \sigma).
\eal

\subsection{Type theory}
\label{subsec:types}
In the following, we review the notion of the type theory. For more detail, see e.g. Ref.~\cite{Csiszar_Korner_2011, cover_1999_elements}.
Let $\mathcal{X}$ be a finite alphabet, and let
$\mathcal{P}(\mathcal{X})$ denote the set of probability
distributions on $\mathcal{X}$.
The type (empirical distribution) of a sequence
$x^n=(x_1,\ldots,x_n)\in\mathcal{X}^n$ is defined by
\begin{equation}
    P_{x^n}(a)
    \coloneqq \frac{1}{n}
    \bigl|\{i\in\{1,\ldots,n\}:x_i=a\}\bigr|,
    \qquad a\in\mathcal{X}.
\end{equation}
We denote the set of types of sequences of length $n$ by $\mathcal{P}_n(\mathcal{X})\coloneqq \left\{P\in\mathcal{P}(\mathcal{X}):nP(a)\in\mathbb{Z}_{\geq 0}~\forall a\in\mathcal{X}\right\}.$
For $P\in\mathcal{P}_n(\mathcal{X})$, its type class is defined as $T^n_P\coloneqq \{x^n\in\mathcal{X}^n:P_{x^n}=P\}$.

Each type is uniquely specified by the integer counts
$(nP(a))_{a\in\mathcal{X}}$, each of which takes a value
in $\{0,\ldots,n\}$. Consequently, it holds that
\begin{equation}\label{eq: type counting}
    |\mathcal{P}_n(\mathcal{X})|
    \leq(n+1)^{|\mathcal{X}|}.
\end{equation}
Thus, for a fixed alphabet $\mathcal{X}$, the number of
types grows at most polynomially in $n$, even though
the number of sequences is $|\mathcal{X}|^n$.
Noting that the cardinality of the type class with $n$-type $P$ is bounded as
\begin{equation}
    (n+1)^{-|\mathcal{X}|}2^{nH(P)}
    \leq |T^n_P|
    \leq 2^{nH(P)}.
    \label{eq:type-class-size}
\end{equation}

Let $\mathcal{Y}$ be another finite alphabet.
Given $P\in\mathcal{P}_n(\mathcal{X})$, a conditional
distribution $V:\mathcal{X}\to\mathcal{P}(\mathcal{Y})$
is a conditional type compatible with $P$ if
\begin{equation}
    nP(a)V(b|a)\in\mathbb{Z}_{\geq 0}
    \qquad
    \text{for all }(a,b)\in\mathcal{X}\times\mathcal{Y}.
\end{equation}
Conditional types that agree on the support of $P$
are identified; the values of $V(\cdot|a)$ for
$P(a)=0$ are immaterial.
We denote the resulting set of conditional types by
$\mathcal{V}_n(\mathcal{Y}|P)$.

For $x^n\in T^n_P$ and
$V\in\mathcal{V}_n(\mathcal{Y}|P)$, the $V$-shell of
$x^n$ is defined as
\begin{equation}
    T_V(x^n)\coloneqq 
    \left\{y^n\in\mathcal{Y}^n:
    \begin{array}{l}
    |\{i:(x_i,y_i)=(a,b)\}|=nP(a)V(b|a)\\
    \text{for all }(a,b)\in\mathcal{X}\times\mathcal{Y}
    \end{array}
    \right\}.
    \label{eq:conditional-type-shell}
\end{equation}

\section{Sphere-packing bound under unitary-invariance (Proof of Theorem~\ref{thm: unitary invariant sphere-packing})}\label{sup sec: sphere-packing bound}
In the following, we will derive the sphere-packing bound for the decoder $\qty{\Lambda_m}_m$ with the unitary invariance $[\Lambda_m, U^{\otimes n}]=0, \forall U\in U(d), \forall m$.
Our goal is to show the following.
\begin{thmboxed}
\begin{thm}[Unitary-invariant sphere-packing bound, Theorem~\ref{thm: unitary invariant sphere-packing} in the main text.]\label{sup thm: unitary invariant sphere-packing}
    Let $P_n$ be the $n$-type from which the codewords are drawn, converging to a fixed probability distribution $P$, and $W:\mX\to\mD(\mH)$ be a classical-quantum channel. Then, for any coding scheme with a unitary-invariant decoder, it holds that 
    \bal\label{sup eq: unitary invariant sphere-packing}
    \limsup_{n\to\infty}-\frac{1}{n}\log P^n_e(P_n, R, W)\leq \sup_{1/2\leq \alpha\leq 1}\frac{\alpha-1}{\alpha}\qty(R-\sI^{\rm Aug}_\alpha(P:W))
    \eal
    for the high-rate regime above a certain critical rate $R_{\rm crit}\leq R$.
    Here, $\sI^{\rm Aug}_\alpha(P:W)$ is the sandwiched Augustion information.
\end{thm}    
\end{thmboxed}

\subsection{Reduction to the pinched hypothesis testing}
First, we reduce the converse bound of the decoding error to the error probability of another information theoretic task called the \emph{quantum hypothesis testing}.
Let us briefly explain the setting of quantum hypothesis testing. Suppose that one is given a quantum state, either $\rho$ or $\sigma$. The goal is to correctly guess which state is given by performing a measurement described by a binary POVM $\qty{M, I-M}$. Here, $M$ corresponds to the measurement outcome inferring that $\rho$ is given, and $I-M$ corresponds to that inferring $\sigma$ is given. In this setting, there are two types of errors: type I error is the case where one is given $\rho$ but infers that $\sigma$ is given, occurring with probability $\alpha=\Tr[\rho (I-M)]$, and type II error corresponds to the other situation occurring with probability $\beta=\Tr[\sigma M]$.

For states $\rho,\tau$ and $0<\mu<1$, define the following quantity.
\begin{equation}
  \alpha_\mu(\rho\Vert\sigma)
  \coloneqq \min\left\{
    \Tr\rho(I-M):0\leq M\leq I,\ \Tr\sigma M\leq\mu
  \right\}.
  \label{eq:binary-alpha}
\end{equation}
Operationally, this quantity represents the minimum type I probability when the probability of type II error is kept smaller than $\mu$. 
The corresponding quantity under the unitary invariance on the test is defined as 
\begin{equation}
  \alpha^U_\mu(\rho\Vert\sigma)
  \coloneqq \min\left\{
    \Tr\rho(I-M):0\leq M\leq I,\ \Tr\sigma M\leq\mu,
    \ [M,U^{\otimes n}]=0,~ \forall U\in U(d)
  \right\}.
  \label{eq:binary-alpha-UI}
\end{equation}

The following lemma allows us to connect the type I error probability under a unitary-invariant test with that without the invariance.
\begin{lemboxed}
  \begin{lem}[Pinching does not change the type I error under unitary-invariance]
For every state $\rho_n\in\mD(\mH^{\otimes n})$, $\sigma\in\mD(\mH)$, and every unitary-invariant test $\Lambda_n$,
\begin{equation}
  \Tr\rho_n\Lambda _n
  =\Tr\mP_{\sigma^{\otimes n}}(\rho_n)\Lambda_n.
  \label{eq:pinching-statistics}
\end{equation}
Consequently,
\begin{equation}
  \alpha^U_\mu(\rho_n\Vert\sigma_n)
  =\alpha^U_\mu(\mP_{\sigma_n}(\rho_n)\Vert\sigma_n)
  \geq
  \alpha_\mu(\mP_{\sigma_n}(\rho_n)\Vert\sigma_n).
  \label{eq:pinching-relaxation}
\end{equation}
\end{lem}  
\end{lemboxed}

\begin{proof}
Due to Schur-Weyl duality
\bal
\mH^{\otimes n}=\bigoplus_{\lambda\in Y^n_d}\mW_\lambda\otimes \mU_\lambda,
\eal
$M_n$ and $\sigma^{\otimes n}$ are decomposed as 
\bal
\Lambda_n=\bigoplus_{\lambda\in Y^n_d}I_{\mW_\lambda}\otimes \Lambda_{n,\lambda},\qquad
\sigma^{\otimes n}=\bigoplus_{\lambda\in Y^n_d}\sigma_{n,\lambda}\otimes I_{\mU_\lambda},
\eal
meaning that $[\Lambda_n,\sigma^{\otimes n}]=0$ holds.
Hence, it holds that $\mP_{\sigma^{\otimes n}}(\Lambda_n)=\Lambda_n$. Self-adjointness of the pinching map gives
\bal
  \Tr\mP_{\sigma^{\otimes n}}(\rho_n)\Lambda_n
  =\Tr\rho_n\mP_{\sigma^{\otimes n}}(\Lambda_n)
  =\Tr\rho_n\Lambda_n.
\eal
This proves the equality in Eq.~\eqref{eq:pinching-relaxation}.  The inequality
follows because the last minimization is over all tests, a superset of the
unitary-invariant tests.
\end{proof}

The following proposition connects the decoding error of the channel coding with the error probability of the hypothesis testing.
\begin{proboxed}
    \begin{pro}\label{pro: lower bound decoding error hypothesis testing}For any classical-quantum channel $W$, any fixed quantum state $\sigma$, and every code with a unitary-invariant decoder, the decoding error probability satisfies
    \bal
    P_{\rm e}^{n}(P_n, R, W)\geq \frac12
  \min_{m\in \mM_n}
  \alpha_{\frac{2}{M_n}}\qty(
    \mP_{\sigma^{\otimes n}}(\rho_{x^n(m)})\middle\Vert\sigma^{\otimes n}).
    \eal
    \end{pro}
\end{proboxed}
\begin{proof}
Let us denote $\ve_m\coloneqq 1-\Tr[\rho_{x^n(m)}\Lambda_m]$, and $\bar\ve\coloneqq 1-\frac{1}{M_n}\Tr[\rho_{x^n(m)}\Lambda_m]$.
Let us define a subset $\mG$ of messages as $\mG_n\coloneqq \qty{m\in\mM_n~|~\ve_m\leq 2\bar\ve}$. Here, due to Markov's inequality, it holds that $\abs{\mG_n}\geq M/2$.

For a given unitary-invariant decoder $\qty{\Lambda_m}_m$, let us denote $\beta_m\coloneqq \Tr[\Lambda_m\sigma^{\otimes n}]$. Since it holds that $\sum_m\beta_m=1$, we have
\bal
\min_{m\in\mG_n}\beta_m\leq \frac{1}{\abs{\mG_n}}\sum_{m\in\mG_n}\beta_m\leq\frac{1}{\abs{\mG_n}}\leq \frac{2}{M_n},
\eal
implying that there exists a message $m^*\in\mM_n$ such that $\Tr[\Lambda_{m^*}\sigma^{\otimes n}]\leq 2/M_n$.

For this message $m^*\in\mM_n$, nothing that the unitary-invariant decoder satisfies $\mP_{\sigma^{\otimes n}}(\Lambda_{m^*})=\Lambda_{m^*}$, we have
\bal
\Tr[\Lambda_{m^*}\rho^{x^n}(m^*)]=\Tr[\Lambda_{m^*}\mP_{\sigma^{\otimes n}}(\rho_{x^n(m^*)})]=\ve_m\leq 2\bar\ve,
\eal
meaning that 
\bal
2\bar\ve\geq \alpha_{2/M_n}(\mP_{\sigma^{\otimes n}}(\rho_{x^n(m^*)})\|\sigma^{\otimes n}),
\eal
which concludes the proof.
\end{proof}

The following lemma relates the regularized divergence of the pinched state with the sandwiched divergence of the original quantum state, which is used in the following discussion.
\begin{lemboxed}
    \begin{lem}[Limit of the pinched divergences]\label{prop: pinching-asymptotics}
Let $\qty{P_n}$ be the sequence of $n$-type converging to a probability distribution $P$, and let $\sigma$ be a full-rank state.  Then for any $x^n\in T^n_{P_n}$ and 
$0<\alpha<1$, it holds that 
\begin{equation}
  \lim_{n\to\infty}\frac1n
  D_\alpha(\mP_{\sigma^{\otimes n}}(\rho_n)\Vert\sigma^{\otimes n})
  =\sum_{x\in\mX}P(x)\widetilde D_\alpha(\rho_x\Vert\sigma).
  \label{eq:pinched-limit-sandwiched}
\end{equation}
\end{lem}
\end{lemboxed}
\begin{proof}
Due to Ref.~\cite[Lemma 3]{Hayashi_2016_correlation_detection}, it holds that 
\bal
D_\alpha(\mP_{\sigma^{\otimes n}}(\rho_{x^n})\Vert\sigma^{\otimes n})\leq \sD_\alpha(\rho_{x^n}\Vert\sigma^{\otimes n})\leq D_\alpha(\mP_{\sigma^{\otimes n}}(\rho_{x^n})\Vert\sigma^{\otimes n})+\log \abs{\spec(\sigma^{\otimes n})}
\eal
Since it holds that $\abs{\spec(\sigma^{\otimes n})}={\rm poly}(n)$ due to the type counting argument in Eq.~\eqref{eq: type counting}, we have
\bal
\lim_{n\to\infty}\frac{1}{n}D_\alpha(\mP_{\sigma^{\otimes n}}(\rho_n)\Vert\sigma^{\otimes n})=\lim_{n\to\infty}\frac{1}{n}\sD_\alpha(\rho_{x^n}\Vert\sigma^{\otimes n}).
\eal
By additivity of the sandwiched R\'enyi divergence, 
\bal
\frac1n\widetilde D_\alpha
  (\rho_{x^n}\Vert\sigma^{\otimes n})
 =\sum_xP_n(x)\widetilde D_\alpha(\rho_x\Vert\sigma),
\eal
holds, which concludes the proof.
\end{proof}

\subsection{Converse bound from large deviation}
\label{sec:triangular-hoeffding}

Fix a faithful state $\sigma$ and a type sequence $P_n\to P$, and let $\rho_n, \sigma_n$ be
$\rho_n\coloneqq \mP_{\sigma_n}(\rho_{x^n})$ and $\sigma_n\coloneqq \sigma^{\otimes n}$ with $x^n\in T^n_{P_n}$.
The two hypotheses commute.  Their eigenvalues in a common eigenbasis
therefore define classical distributions $p_n$ and $q_n$ on $\mZ\coloneqq \mX^n$, respectively,
with $q_n>0$.  Their normalized log-moment function is
\begin{equation}
 \psi_{n,\sigma}(\alpha)
 \coloneqq \frac1n\log\Tr\rho_n^\alpha\sigma_n^{1-\alpha}
 =\frac1n\log\sum_{z\in\mZ}p_n(z)^\alpha q_n(z)^{1-\alpha},
 \qquad 0<\alpha<1.
 \label{eq:psi-n-classical}
\end{equation}

We define the single-letter sandwiched log-moment functions by
\bal
 \widetilde\psi_{x,\sigma}(\alpha)
 \coloneqq \log\Tr\left(
 \sigma^{\frac{1-\alpha}{2\alpha}}\rho_x
 \sigma^{\frac{1-\alpha}{2\alpha}}\right)^\alpha
 =(\alpha-1)\widetilde D_\alpha(\rho_x\Vert\sigma).
\eal
Lemma~\ref{prop: pinching-asymptotics} gives 
\begin{equation}
 \psi_{n,\sigma}(\alpha)\xrightarrow{n\to\infty}
 \psi_{\sigma}(\alpha)
 \coloneqq \sum_xP(x)\widetilde\psi_{x,\sigma}(\alpha)
 =(\alpha-1)J_\alpha(P,W\mid\sigma),
 \label{eq:psi-sigma}
\end{equation}
where
\bal
 J_\alpha(P,W\mid\sigma)
 \coloneqq \sum_xP(x)\widetilde D_\alpha(\rho_x\Vert\sigma).
\eal
Thus, the limit of the functions $\psi_{n,\sigma}(\alpha)$ is the convex
combination of the single-letter sandwiched log-moment functions.
The displayed pinching correction controls the difference at finite $n$.

Each $\widetilde\psi_{x,\sigma}$ is finite and real analytic on $(0,1)$, even when $\rho_x$ is singular.  Indeed, the trace can be evaluated as
\bal
 \Tr\left(\rho_x^{1/2}\sigma^{(1-\alpha)/\alpha}\rho_x^{1/2}\right)^\alpha
\eal
on $\supp \rho_x$, where the operator inside the power is positive
definite.  Hence $\psi_\sigma$ is finite and differentiable on $(0,1)$;
it is convex, as is each $\psi_{n,\sigma}$.  These are the log-moment conditions
needed for the G\"artner--Ellis argument below.  The sufficiency of this
open order interval for the Hoeffding argument is also noted in
\cite[Remark~4.7]{Hiai_hypothesis_testing_group_2009}.

We define a function $ H_\sigma(P,r)$ as 
\begin{equation}
 H_\sigma(P,r)
 \coloneqq \sup_{0<\alpha<1}
 \frac{-(1-\alpha)r-\psi_\sigma(\alpha)}{\alpha}
 =\sup_{0<\alpha<1}\frac{1-\alpha}{\alpha}
 \left[J_\alpha(P,W\mid\sigma)-r\right].
 \label{eq:H-sigma}
\end{equation}

\begin{proboxed}
\begin{pro}[Hoeffding exponent for the pinched states]
\label{prop:triangular-Hoeffding}
Let $\sigma$ be a full-rank state, and $\qty{P_n}_n$ be the sequence of $n$-types converging to a probability distribution $P$, and $R>0$. Then, it holds that 
\begin{equation}
 \lim_{n\to\infty}-\frac1n\log
 \alpha_{\mu_n}(\rho_n\Vert\sigma_n)
 =H_\sigma(P,R).
 \label{eq:triangular-Hoeffding-conclusion}
\end{equation}
\end{pro}
\end{proboxed}
\begin{proof}
Due to the discussion above, the sequence of the log-moment functions $ \qty{\psi_{n,\sigma}(\alpha)}_n$ converges to $\psi_{\sigma}(\alpha)$, which is finite, convex, and differentiable on $(0,1)$.
Furthermore, one can check that the left-hand derivative $\partial_-\psi_{\sigma}(1)$ of the log-moment function satisfies
\bal
\partial_-\psi_{\sigma}(1)&=\lim_{\alpha\to 1}\sum_{x\in\mX} P_n(x)\sD_\alpha(\rho_x\|\sigma)\\
&=\sum_{x\in\mX} P(x)D(\rho_x\|\sigma)=\lim_{n\to\infty}\frac{1}{n}D(\rho_n\|\sigma_n).
\eal

This allows us to employ the previous results about the error exponent of hypothesis testing with correlated hypotheses~\cite {hiai_error_correlated_2008, Hiai_hypothesis_testing_group_2009} (see also Ref.~\cite{Lipka-Bartosil_quantum_dichotomies}), which yields
\bal
 \lim_{n\to\infty}-\frac1n\log
 \alpha_{\mu_n}(\rho_n\Vert\sigma_n)
 =\sup_{0<\alpha\le1}
   \frac{-(1-\alpha)r-\psi_\sigma(\alpha)}{\alpha}=H_\sigma(P,R).
\eal

\end{proof}

Combining Proposition~\ref{prop:triangular-Hoeffding} with the lower bound on the decoding error with respect to the type I error in Proposition~\ref{pro: lower bound decoding error hypothesis testing}, we immediately have the following.

\begin{proboxed}
   \begin{pro}[Fixed-auxiliary-state converse]
\label{pro: fixed-sigma-converse}
Let $R>0$ be a target rate, $\qty{P_n}_n$ be a sequence of $n$-types satisfying $P_n\to P$ and $\sigma$ be a full-rank state.
Then, every sequence of constant-composition codes with unitary-invariant decoders satisfies
\begin{equation}
  \limsup_{n\to\infty}-\frac1n\log P_{\rm e}^{n}(P_n, R, W) \leq H_\sigma(P,R).
  \label{eq:fixed-sigma-sphere-packing}
\end{equation}
\end{pro} 
\end{proboxed}

\subsection{Sandwiched sphere-packing bound via minimax}\label{sec:high-rate-minimax}
Even though we already have a converse bound through an entropic quantity, the statement of Theorem~\ref{sup thm: unitary invariant sphere-packing} does not immediately hold from Proposition~\ref{pro: fixed-sigma-converse}. Indeed, if we compare the right-hand side of Eq.~\eqref{eq:fixed-sigma-sphere-packing} with the right-hand side of Theorem~\ref{sup thm: unitary invariant sphere-packing}, we can see that
\bal
\inf_{\sigma>0} H_\sigma(P,R)&= \inf_{\sigma>0}\sup_{0<\alpha<1}\frac{1-\alpha}{\alpha}
\qty[J_\alpha(P,W\mid\sigma)-R]\\
&\geq \sup_{0<\alpha<1}\inf_{\sigma>0}\frac{1-\alpha}{\alpha}
\qty[J_\alpha(P,W\mid\sigma)-R]\\
&\geq \sup_{0<\alpha<1}\inf_{\sigma\in\mD(\mH)}\frac{1-\alpha}{\alpha}\qty[J_\alpha(P,W\mid\sigma)-R]\\
&=\sup_{0<\alpha<1}\frac{\alpha-1}{\alpha}\qty[R-\sI^{\rm Aug}_\alpha(P;W)].
\eal
Therefore, to obtain the sphere-packing bound with the sandwiched Rényi Augustin information, we need to justify the exchange of the inf and sup, and also extend the range of taking the infimum.
To this end, we define the objective function of the optimization as 
\begin{equation}
 F_R(s,\sigma)
 \coloneqq s\left[J_{1/(1+s)}(P,W\mid\sigma)-R\right],
 \qquad s>0.
 \label{eq:minimax-F}
\end{equation}
We put $F_R(0,\sigma)\coloneqq 0$ when taking the supremum over
$0\le s\le1$.

\begin{proboxed}
\begin{pro}[Minimax identity]
\label{prop:truncated-sion}
For any classical-quantum state $W:x\mapsto \rho_x$, it holds that 
\begin{equation}
 \inf_{\sigma\in\mD(\mH)}\sup_{0\le s\le1}F_R(s,\sigma)
 =\sup_{0\le s\le1}\inf_{\sigma\in\mD(\mH)} F_R(s,\sigma).
 \label{eq:truncated-sion}
\end{equation}
Here, we allow the quantities to diverge.
Moreover, the infimum and supremum are attained, and their optimizers form a saddle point.
\end{pro}    
\end{proboxed}

\begin{proof}
For a fixed $\sigma$, concavity of $s\mapsto F_R(s,\sigma)$ on
$s>0$ follows from the concavity of the sandwiched auxiliary
function \cite[Appendix~B]{mosonyi_strong_cq_channel_2017}.  For $0<s\le1$,
the map $\sigma\mapsto F_R(s,\sigma)$ is convex and lower
semicontinuous by the sandwiched divergence's convexity for
$\alpha\ge1/2$ \cite[Theorem~2]{Frank_2013}.
Then, due to  Ref.~\cite[Lemma~II.3]{Mosonyi_2021_divergence_radii}, we have 
\bal
\min_{\sigma\in\mD(\mH)}\sup_{0< s\le1}F_R(s,\sigma)
 =\sup_{0< s\le1}\min_{\sigma\in\mD(\mH)} F_R(s,\sigma).
\eal
Indeed, $\inf$ on both sides can be replaced with $\min$. We denote the minimizer of the left-hand side as $\sigma^*\in\mD(\mH)$.

Let us see that including $s=0$ in the range of the supremum does not change the optimized values on both sides.
As for the left-hand side, the nonnegativity of $ J_{1/(1+s)}(P,W\mid\sigma)\coloneqq \sum_xP(x)\widetilde D_{1/(1+s)}(\rho_x\Vert\sigma)$ implies $F_R(s,\sigma)\geq -sR$, meaning that $\sup_{0<s<1}F_R(s,\sigma)\geq 0$. Since we fixed $F_R(0, \sigma)=0$, the added point $s=0$ does not contribute to the optimal value of the left-hand side. 
As for the right-hand side, we denote $B(s)\coloneqq \inf_{\sigma\in\mD(\mH)}F_R(s,\sigma)$. Let us take a full-rank $\sigma$ and take the limit $s\to 0$; then we can upper bound $\limsup_{s\to+0}B(s)\leq 0$. Combining this with the lower bound $B(s)\geq 0$ yields $\lim_{s\to+ 0}B(s)=0$.
From this, adjoining $s=0$ does not change the right-hand side either.

Finally, let us show that there is a saddle point.
The function $B(s)-sR$ is upper semicontinuous on $(0,1]$ as the
infimum of fixed-state continuous functions
\cite[Lemma~III.3]{mosonyi_strong_cq_channel_2017}. Furthermore, since we see that $\lim_{s\to +0}B(s)=0=B(0)$, $B(s)$ is upper semicontinuous in the compact interval $s\in[0,1]$. From this, we can see that the supremum on the right-hand side is attained at some $s^*\in[0,1]$.
From the observation above, we have
\bal
F_R(s,\sigma^*)\leq F_R(s^*,\sigma^*)\leq F_R(s^*,\sigma),
\eal
meaning that the pair $(s^*,\sigma^*)$ is the saddle point.
\end{proof}

Let us remark that, due to Ref.~\cite[Lemmas~III.9 and III.11]{Mosonyi_2021_divergence_radii}, $ \supp\sigma^*=K\coloneqq \supp\sum_{x\in\mX}p(x)\rho_x$ holds.

\subsubsection{The right-hand side: Sandwiched Augustin information}
Let us see the connection between the quantity on the right-hand side of Eq.~\eqref{eq:truncated-sion} and the sandwiched Augustin information.
Since $\sigma\mapsto F_R(s,\sigma)$ is not convex in general for $s>1$, one can apply the minimax theorem only for the interval $0\leq s\leq 1$, which restricts the applicable communication rate region as follows:
Let us define
\begin{equation}
 A(s)\coloneqq s\widetilde I^{\rm Aug}_{1/(1+s)}(P,W),
 \qquad
 R_{\mathrm{crit}}(P)\coloneqq A'_-(1),
 \qquad
 R_{\max}(P)\coloneqq A'_+(0).
 \label{eq:critical-rates}
\end{equation}
Here, $A'_-(1)$ is the right-hand deriverate of $A(s)$, and $A'_+(0)$ is the left-hand deriverate of $A(s)$
As the infimum of the fixed-state concave functions, $A$ is concave.
Continuity at order one and the order-one Augustin identity
\cite[Proposition~5]{Cheng_properties_2022} give
\bal
 R_{\max}(P)=\lim_{\alpha\to 1}\sI_\alpha^{\rm Aug}(P;W)=I(P;W).
\eal

If we choose the target rate $R$ so that  it satisfies $R_{\mathrm{crit}}(P)<R<I(P;W)$,
the right derivative of $A(s)-sR$ at zero is
$R_{\max}(P)-R>0$, while its left derivative at one is
$R_{\mathrm{crit}}(P)-R<0$.  Concavity excludes both endpoints;
thus every maximizer over $[0,1]$ lies in $(0,1)$.

Choose a saddle point $(s_*,\sigma_*)$ with $0<s_*<1$, and set
$\alpha_*\coloneqq 1/(1+s_*)\in(1/2,1)$. 
From the definition of the saddle point $F_R(s^*,\sigma)\geq F_R(s^*,\sigma^*)\geq F_R(s,\sigma^*)$, we can see that $\sigma^*$ minimizes $J_{\alpha^*}(P,W\mid\sigma)$.

\subsubsection{The left-hand side: Hoeffding divergence}
The result of Proposition~\ref{prop:truncated-sion} can actually be strengthened: Indeed, the following also holds.
\bal
 \inf_{\sigma\in\mD(\mH)}\sup_{0\le s\le1}F_R(s,\sigma)
 =\sup_{s\geq 0}\inf_{\sigma\in\mD(\mH)} F_R(s,\sigma).
\eal
This can be verified as follows:
For the minimizer $\sigma^*$, define a function $G_R(s)\coloneqq F_R(s,\sigma_*)$. Noting that $s\widetilde D_{1/(1+s)}(\rho_x\Vert\sigma^*)
 =-(1+s)\widetilde \psi_{1/(1+s)}(\rho_x\Vert\sigma^*)$, and following the discussion in \cite[Corollary~B.2]{mosonyi_strong_cq_channel_2017}, we can see that $s\widetilde D_{1/(1+s)}(\rho_x\Vert\sigma^*)$ is concave on
$s>0$ for any $x\in\mX$.  Hence $G_R$, is also concave on $s>0$, implying that 
\begin{equation}
H_{\sigma^*}(P,R)=\sup_{s\ge0}F_R(s,\sigma^*)=F_R(s^*,\sigma^*).
 \label{eq:fixed-center-full-maximum}
\end{equation}
holds. 
Also, the maximizer $s^*\in(0,1)$ of the function $G_R(s)$ is unique, which can be verified as follows: Noting that $G_R(s)$ is real analytic and concave, if there are multiple maximizers, $G_R(s)$ needs to be a constant over $s\in(0,\infty)$, which is not true because we have
\bal
 \lim_{s\to +0}\frac{G_R(s)}s
 =J_1(P,W\mid\sigma_*)-R
 \ge R_{\max}^{*,A}(P)-R>0.
\eal

The following lemma guarantees that we can first perturb the maximizer $\sigma^*$ to $\sigma_\delta\coloneqq (1-\delta)\sigma^*+\delta I/d$, and then take the limit $ \delta\to 0$ to make the second argument full-rank
\begin{lemboxed}
    \begin{lem}\label{lem:no-optimizer-escape}
    Let $R>0$ be the target rate.
Set $\sigma_\delta\coloneqq (1-\delta)\sigma^*+\delta I/d$ for $\delta>0$.
Then, it holds that 
\begin{equation}
 \lim_{\delta\downarrow0}\sup_{s\ge0}F_R(s,\sigma_\delta)
 =F_R(s^*,\sigma^*).
 \label{eq:faithful-approximation-limit}
\end{equation}
\end{lem}
\end{lemboxed}
\begin{proof}
Let $G_R(s)\coloneqq F_R(s,\sigma_*)$. By the preceding argument, $G_R$ is
concave and differentiable on $(0,\infty)$, and $s_*\in(0,1)$ is its
unique global maximizer. Choose $0<a<s_*<b<1$. Concavity and uniqueness
give
\bal
 G_R'(a)
 \ge \frac{G_R(s_*)-G_R(a)}{s_*-a}>0,
 \qquad
 G_R'(b)
 \le \frac{G_R(b)-G_R(s_*)}{b-s_*}<0.
\eal
We will show that these strict derivative inequalities persist when
$\sigma_*$ is replaced by $\sigma_\delta$ and the rate is varied in a
fixed neighborhood of $R$.

To this end, let $K\coloneqq \supp(\sigma^*)$.
Due to the discussion above, any $\rho_x$ with $P(x)>0$ is supported on $K$, and $S^*\coloneqq \sigma^*|_K$ is positive definite. Moreover, $\sigma_\delta$
is block diagonal with respect to $K\oplus K^\perp$, and its
restriction to $K$ is
\bal
 S_\delta\coloneqq (1-\delta)S^*+\frac{\delta}{d}I_K.
\eal
For $0\le\delta\le1/2$, we have
$S_\delta\ge \frac12\lambda_{\min}(S_*)I_K$.
Thus the restrictions remain uniformly positive definite as
$\delta\to +0$, even though $\sigma^*$ may be singular on the
full output space.

For each $x\in \supp P$, regard $\rho_x$ as an operator on
$K$, let $L_x\coloneqq \supp\rho_x$, and define
\bal
 B_x(s,\delta)
 \coloneqq \left.\rho_x^{1/2}S_\delta^s\rho_x^{1/2}\right|_{L_x}.
\eal
This operator is positive definite on $L_x$, including at $\delta=0$.
The nonzero eigenvalues of
$\sigma_\delta^{s/2}\rho_x\sigma_\delta^{s/2}$ are precisely those of
$B_x(s,\delta)$. Consequently, the definition of the sandwiched
R\'enyi divergence yields
\bal
 F_R(s,\sigma_\delta)
 =-(1+s)\sum_{x\in\operatorname{supp}P}P(x)
   \log\operatorname{Tr}_{L_x}
   \left[B_x(s,\delta)^{1/(1+s)}\right]-sR.
\eal
Matrix powers depend smoothly on their exponent and on a positive
definite matrix. Hence the expression on the right and its first
$s$-derivative are jointly continuous on
$[a,b]\times[0,1/2]$. Since this set is compact, it follows that
\bal
 \sup_{s\in[a,b]}
 \left(
   \left|F_R(s,\sigma_\delta)-G_R(s)\right|
   +\left|\partial_sF_R(s,\sigma_\delta)-G_R'(s)\right|
 \right)
 \longrightarrow0
 \qquad(\delta\to +0).
\eal

Now set $m\coloneqq \min\{G_R'(a),-G_R'(b)\}>0$.
By the uniform convergence of the derivatives, there exists
$\delta_0\in(0,1/2)$ such that, for every $0<\delta<\delta_0$,
\bal
 \partial_sF_R(a,\sigma_\delta)>\frac{3m}{4},
 \qquad
 \partial_sF_R(b,\sigma_\delta)<-\frac{3m}{4}.
\eal
Choose $0<\eta<\min\{m/4,R/2\}$ and let
$I\coloneqq (R-\eta,R+\eta)$. Since
\bal
 \partial_sF_r(s,\sigma_\delta)
 =\partial_sF_R(s,\sigma_\delta)-(r-R),
\eal
we obtain, simultaneously for every $r\in I$ and
$0<\delta<\delta_0$,
\bal
 \partial_sF_r(a,\sigma_\delta)>\frac m2,
 \qquad
 \partial_sF_r(b,\sigma_\delta)<-\frac m2.
\eal

Fix such $r$ and $\delta$, and write $f(s)\coloneqq F_r(s,\sigma_\delta)$.
The function $f$ is concave on $(0,\infty)$ and extends continuously
to $s=0$ with $f(0)=0$. A differentiable concave function lies below
each of its tangent lines. Therefore, for $0\le s<a$,
\bal
 f(s)\le f(a)+f'(a)(s-a)<f(a),
\eal
whereas, for $s>b$,
\bal
 f(s)\le f(b)+f'(b)(s-b)<f(b).
\eal
It follows that no point outside $[a,b]$ can be a global maximizer.
As $f$ is continuous on $[a,b]$, its global maximum is attained there,
and thus
\bal
 \sup_{s\ge0}F_r(s,\sigma_\delta)
 =\max_{s\in[a,b]}F_r(s,\sigma_\delta).
\eal

Applying this identity at $r=R$ and using the uniform convergence
proved above, we conclude that
\begin{align*}
 &\left|
   \sup_{s\ge0}F_R(s,\sigma_\delta)-F_R(s_*,\sigma_*)
  \right|\\
 &\quad=
  \left|
   \max_{s\in[a,b]}F_R(s,\sigma_\delta)
   -\max_{s\in[a,b]}G_R(s)
  \right|\\
 &\quad\le
  \sup_{s\in[a,b]}
  \left|F_R(s,\sigma_\delta)-G_R(s)\right|
  \xrightarrow{\delta\to +0}0.
\end{align*}

\end{proof}

Combining these, we reach the main statement as follows:
\begin{proof}[{Proof of Theorem~\ref{sup thm: unitary invariant sphere-packing}}]
For every sufficiently small $\delta>0$, the state $\sigma_\delta$ of
Lemma~\ref{lem:no-optimizer-escape} is faithful on the full output space.
Due to Proposition~\ref{pro: fixed-sigma-converse}, it holds that
\bal
\limsup_{n\to\infty}-\frac1n\log P_{\rm e}^{n}(P_n, R_n, W) \le H_{\sigma_\delta}(P,R)=\sup_{s\ge0}F_R(s,\sigma_\delta), \qquad\forall\delta>0.
\eal
For the high-rate regime $R_{\mathrm{crit}}(P)<R<I(P;W)$, combining Lemma~\ref{lem:no-optimizer-escape} and Proposition~\ref{prop:truncated-sion} and taking the limit $\delta\to + 0$ yields
\bal
\lim_{\delta\to +0}\sup_{s\ge0}F_R(s,\sigma_\delta)&=F(s^*,\sigma^*)\\
&=\sup_{0\le s\le1}\inf_{\sigma\in\mD(\mH)} F_R(s,\sigma)\\
&=\sup_{0\le s\le1}s\qty[
 \widetilde I^A_{1/(1+s)}(P,W)-R].
\eal
Recalling that $\alpha=\frac{1}{1+s}$, we reach the claim.

\end{proof}

\section{Proof of Theorem~\ref{thm: universal random coding bound}}\label{sup sec: random coding}
In the following, we will show that the random-coding bound with the sandwiched Augustin information can be achievable universally. Our statement is as follows:
\begin{thmboxed}
    \begin{thm}[Theorem~\ref{thm: universal random coding bound} in the main text]\label{supthm: universal random coding bound}
    Let $P$ be a probability distribution, and $R$ be such that $0\leq R\leq H(P)$. There exists a fixed sequence of the codebook with constant composition $\qty{P_n}_n$ converging to $P$, and the decoder such that the decoding error satisfies 
    \bal
    \liminf_{n\to\infty}-\frac{1}{n}\log P^n_e(P_n, R, W)\geq \sup_{1/2\leq \alpha\leq 1}\frac{\alpha-1}{\alpha}\qty(R-\sI^{\rm Aug}_\alpha(P:W))
    \eal
    for any classical-quantum channel $W$.
\end{thm}
\end{thmboxed}

The proof consists of three steps:
\begin{enumerate}
    \item Under a fixed good codebook, we bound the decoding error probability for a decoder constructed from a certain family of positive definite matrices associated with the codewords.
    \item Showing that we can bound the error probability derived in the previous step by the measured Rényi Augustin information up to the polynomial additional overhead, by a channel-independent decoder.
    \item Observing that the regularized version of the measured Augustin information coincides with the sandwiched Augustin information, we reach the claim.
\end{enumerate}
In Ref.~\cite{Beigi_2024_lower_bound}, they prove that the random-coding bound with the sandwiched Rényi mutual information is achievable in the channel-aware scenario.
In our case, we need to choose a fixed codebook and a universal decoder that work independently of the given channel.

We begin by constructing the encoder and decoder that achieve the measured Augustin random-coding bound, as follows, and then reduce it to the sandwiched Augustin random-coding bound.

\subsection{Construction of the codebook}
We first discuss how we construct the codebook $\mC_n\subset T^n_P$.
The existence of the following sets is key of the construction.
\begin{defnboxed}
    \begin{defn}[$(P, R, \delta)$-good set for $x^n$, see e.g. Ref.~{\cite[Definition 1]{Matsuura_2025_universal_resolvability}}]
        For a $n$-type $P$, positive number $R<H(X)_P$, $\delta>0$, a set $\mL_n\subset T^n_P$ of length-$n$ sequences is called $(P, R, \delta)$-good set for $x^n$ if 
        \bal
        \abs{\mL_n}&=e^{nR-\delta},\\
        \abs{T_V(x^n)\cap\mL_n}&\leq \abs{T_V}e^{-n(H(X)_P-R)}
        \eal
        for any conditional $n$-type $V$. Here, $T_V(x^n)$ is $V$-shell with respect to $x^n$.
    \end{defn}
\end{defnboxed}

\begin{defnboxed}
    \begin{defn}[$(P, R, \delta)$-good codebook, see e.g. Ref.~{\cite[Definition 1]{Matsuura_2025_universal_resolvability}}]
        A set $\mM_n\subset T^n_P$ is a $(P, R, \delta)$-good codebook if, for any $x^n\in\mM_n$, the set $\mM_n\backslash\qty{x^n}$ is a $(P,R,\delta)$-good set of $x^n$. 
    \end{defn}
\end{defnboxed}
In fact, one can find a $(P,R,\delta)$-codebook, which is also employed in the previous constructions of the universal classical-quantum channel coding scheme~\cite{Hayashi_2009_universal_coding, Matsuura_2025_universal_resolvability}.

\begin{lemboxed}
    \begin{lem}[Lemma 10.1 of Ref.~\cite{Csiszar_Korner_2011}]\label{lem: existence of codebook csiszar}
        For a type $P$
        For a probability distribution $P$, a sequence $\qty{P_n}_n$ of $n$-types satisfying $P_n\to P$, and positive number $R>0$, there exists a sequence $\qty{\mC_n}_{n\geq N}$ of codebook for a sufficiently large $N$ such that every $\mC_n$ is a $(P_n, R, \delta_1(n))$-good codebook for $\delta_1(n)=\Omega((\log n)^{\aX^2})$
    \end{lem}
\end{lemboxed}

\begin{lemboxed}
    \begin{lem}[Lemma 6 of Ref.~{\cite{Matsuura_2025_universal_resolvability}}]\label{lem: packing}
        Let $\mL_n$ be a $(P_n, R, \delta)$- good set for $x^n$. Let $p_{\mL_n}$ be a distribution on $\mX^n$ defined as 
        \begin{equation}
            p_{\mL_n}(y^n)=
            \left\{ \,
    \begin{aligned}
    &\abs{\mL_n}^{-1}~~~(y^n\in\mL_n)\\
    &0~~~(\mbox{otherwise})
    \end{aligned}
\right.
        \end{equation}
        Let $\mfS_{x^n}\subset\mfS_n$ be a stabilizer of $x^n$. Then, for any $y^n\in T^n_{P_n}\backslash\qty{x^n}$, we have
        \bal
        \frac{1}{\abs{\mfS_{x^n}}}\sum_{\pi\in\mfS_{x^n}}p_{\mL_n}(\pi(y^n))\leq e^{-n H(X)_{P_n}+\delta}.
        \eal
    \end{lem}
\end{lemboxed}
For any $n$, we take the sequence $\mC_n$ of $(P_n, R,\delta_1(n))$-good codebooks as our codebook.

\subsection{Construction of the decoder}
In the following, we explain how to construct the decoder to achieve the random-coding bound with the measured Augustin information.

\subsubsection{Quotient decoder}
Here, we define the quotient operator considered in Ref.~\cite{Beigi_2024_lower_bound, Cheng_operator_layer_cake,Matsuura_2025_universal_resolvability}: for operators $A\geq 0$ and $B>0$
\bal
\frac{A}{B}\coloneqq D\log (B)[A]=\int_0^\infty\dd \lambda (B+\lambda I)^{-1}A(B+\lambda I)^{-1}.
\eal
The quotient operator satisfies the following~\cite{Beigi_2024_lower_bound}:
\begin{enumerate}
    \item $\frac{A}{B}$ is positive semidefinite.
    \item $\frac{A+B}{C}=\frac{A}{C}+\frac{B}{C}$.
    \item $\frac{A}{A+B}\leq \frac{A}{B}$.
    \item $U\frac{A}{B}U^{\dagger}=\frac{UAU^\dagger}{UBU^\dagger}$
\end{enumerate}

Let us consider a family $\qty{S_{x^n}}_{x^n\in \mC_n}$ of positive semidefinite operators on the output Hilbert space $\mH^{\otimes n}$, each of which is associated with the codewords in the codebook $\mC_n$.
From this, one can define the family of operators $\qty{\Lambda_m}_{m\in\mM_n}$ as 
\bal
\Lambda_m\coloneqq \frac{S_{x^n(m)}}{\sum_{m'=1}^{M_n}S_{x^n(m')}},~~m=1,\ldots, M_n,
\eal
the family $\qty{\Lambda_m}_m$ forms a complete set of POVM.
The following lemma gives us a way to obtain an upper bound for the decoding error when the family of positive operators satisfies a certain covariance.

\begin{proboxed}
    \begin{pro}\label{pro: quotient error bound}
        Let $\mC_n$ be a $(P, R, \delta_1(n))$-codebook in Lemma~\ref{lem: existence of codebook csiszar}, and $\qty{S_{x^n}}_{x^n\in T^n_P}$ be a family of the positive definite operators on $\mH^{\otimes n}$ satisfying the covariance condition
        \bal
        U_\pi S_{x^n} U^{\dagger}_\pi=S_{\pi(x^n)}.
        \eal
        Furthermore, suppose that the family $\qty{S_{x^n}}_{x^n\in T^n_P}$ satisfies
        \bal 
        \|\overline S\|\leq 1,
        \qquad 
        \overline S\coloneqq \frac{1}{\abs{T_P^n}}\sum_{x^n\in T_P^n}S_{x^n}.
        \eal
        Then, for sufficiently large $n$ and any codeword $x^n(m)\in\mC_n$, it holds that
        \bal\label{Eq: quotient error bound}
        1-\Tr[\rho_{x^n(m)}\Lambda_{m}]\leq (M_n-1)^se^{s\delta_1(n)}\Tr[\rho_{x^n(m)}S_{x^n(m)}^{-s}],\qquad\forall s\in(0,1]
        \eal
    \end{pro}
\end{proboxed}
\begin{proof}

Let us denote $B_m\coloneqq \sum_{x^n\in\mC_n\backslash\qty{x^n(m)}}S_{x^n}$.
Then, it holds that 
\bal
I-\Lambda_m=I-\frac{S_{x^n(m)}}{S_{x^n(m)}+B_m}=\frac{B_m}{S_{x^n(m)}+B_m}\leq \frac{B_m}{S_{x^n(m)}}
\eal
Here, the last inequality is due to the property of the quotient operator.
Since $0\leq I-\Lambda_m\leq I$ and the operator monotonicity of $x\mapsto x^s$ with $0\leq s\leq 1$, we have
\bal
I-\Lambda_m \leq \qty(I-\Lambda_m)^s\leq \qty(\frac{B_m}{S_{x^n(m)}})^s.
\eal

Since $\rho_{x^n(m)}$ and $S_{x^n(m)}$ are invariant under the stabilizer $\mfS_{x^n(m)}$, we have
\bal
1-\Tr[\rho_{x^n(m)}\Lambda_m ]&\leq \Tr\qty[\rho_{x^n(m)}\qty(\frac{B_m}{S_{x^n(m)}})^s]\\
&=\frac{1}{\abs{\mfS_{x^n(m)}}}\sum_{\pi\in \mfS_{x^n(m)}}\Tr\qty[U_\pi\rho_{x^n(m)}U_\pi^\dagger\qty(\frac{B_m}{S_{x^n(m)}})^s]\\
&=\Tr\qty[\rho_{x^n(m)}\frac{1}{\abs{\mfS_{x^n(m)}}}\sum_{\pi\in \mfS_{x^n(m)}}U^\dagger_\pi\qty(\frac{B_m}{S_{x^n(m)}})^s U_\pi]
\eal

Since $x\mapsto x^s$ is operator concave, Jensen's inequality gives
\bal\label{eq: eq so far direct}
\Tr\qty[\rho_{x^n(m)}\frac{1}{\abs{\mfS_{x^n(m)}}}\sum_{\pi\in \mfS_{x^n(m)}}U^\dagger_\pi\qty(\frac{B_m}{S_{x^n(m)}})^s U_\pi]&\leq \Tr\qty[\rho_{x^n(m)}\qty(\frac{1}{\abs{\mfS_{x^n(m)}}}\sum_{\pi\in \mfS_{x^n(m)}}U^\dagger_\pi\frac{B_m}{S_{x^n(m)}} U_\pi)^s]\\
&=\Tr\qty[\rho_{x^n(m)}\qty(\frac{1}{\abs{\mfS_{x^n(m)}}}\sum_{\pi\in \mfS_{x^n(m)}}\frac{U^\dagger_\pi B_m   U_\pi}{S_{x^n(m)}})^s]\\
&=\Tr\qty[\rho_{x^n(m)}\qty(\frac{\overline B_m  }{S_{x^n(m)}})^s]
\eal
where
\bal
 \overline B_m\coloneqq \frac1{\abs{\mfS_{x^n(m)}}}\sum_{\pi\in \mfS_{x^n(m)}}U^\dagger_\pi B_m U_\pi.
\eal
Here, we used the property of the quotient and the invariance of $S_{x^n(m)}$ under the stabilizer $\mfS_{x^n(m)}$.

Now, let us obtain the bound for $\overline B_m$. Utilizing the probability distribution $p_{\mC_n\backslash\qty{x^n(m)}}(y^n)$ on $\mX^n$, defined in Lemma~\ref{lem: packing}, we have
\bal
\overline B_m&=\frac{1}{\abs{\mfS_{x^n(m)}}}\sum_{\pi\in \mfS_{x^n(m)}}\sum_{x^n\in\mC_n\backslash\qty{x^n(m)}}U^\dagger_{\pi}S_{x^n(m)}U_{\pi}\\
&=\frac{M_n-1}{\abs{\mfS_{x^n(m)}}}\sum_{\pi\in \mfS_{x^n(m)}}\sum_{y^n\in T^n_P\backslash\qty{x^n(m)}}p_{\mC_n\backslash\qty{x^n(m)}}(y^n)U^\dagger_{\pi}S_{y^n}U_{\pi}\\
&=\frac{M_n-1}{\abs{\mfS_{x^n(m)}}}\sum_{\pi\in \mfS_{x^n(m)}}\sum_{y^n\in T^n_P\backslash\qty{x^n(m)}}p_{\mC_n\backslash\qty{x^n(m)}}(y^n)S_{\pi^{-1}(y^n)}\\
&=\frac{M_n-1}{\abs{\mfS_{x^n(m)}}}\sum_{\pi\in \mfS_{x^n(m)}}\sum_{y^n\in T^n_P\backslash\qty{x^n(m)}}p_{\mC_n\backslash\qty{x^n(m)}}(\pi(y^n))S_{y^n}.
\eal
Employing Lemma~\ref{lem: packing}, we have
\bal
\frac{M_n-1}{\abs{\mfS_{x^n(m)}}}\sum_{\pi\in \mfS_{x^n(m)}}\sum_{y^n\in T^n_P\backslash\qty{x^n(m)}}p_{\mC_n\backslash\qty{x^n(m)}}(\pi(y^n))S_{y^n}&\leq (M_n-1)\sum_{y^n\in T^n_P\backslash\qty{x^n(m)}}e^{-nH(X)_P+\delta(n)}S_{y^n}\\
&\leq (M_n-1)e^{-nH(X)_P+\delta(n)}\abs{T^n_P}\overline S\\
&\leq (M_n-1)e^{\delta_1(n)}\overline S.
\eal
In the last inequality, we used $\abs{T^n_P}\leq e^{nH(X)_P}$.
Combining the discussions above, we have 
\bal
\overline B_m\leq (M_n-1)e^{\delta_1(n)}\overline S \leq  (M_n-1)e^{\delta_1(n)}I.
\eal
Here, the last inequality follows from the assumption $\|\overline S\|\leq 1$.
Since the quotient preserves the Löewner order of the numerator, it follows that
\bal
 \frac{\overline B_m}{S_x}\leq  (M-1) e^{\delta_1(n)}\frac{I}{S_x}
 =(M-1) e^{\delta_1(n)}S_x^{-1}.
\eal
Plugging this into Eq.~\eqref{eq: eq so far direct} and noting that $x\mapsto x^s$ for $s\in(0,1]$ is operator monotone, we reach the claim.
\end{proof}

\subsubsection{Measured Rényi Augustin information and its duality}
Next, we show that, for a family of positive-definite operators depending on the channel, we can bound the decoding error probability by the measured Rényi Augustin information.
To this end, we first exhibit the variational form of the measured Rényi Augustin information.

\begin{thmboxed}
\begin{thm}[Dual form of the measured Augustin information]\label{thm: measured-augustin duality}
Let $\alpha\in[1/2,1)$. For every finite ensemble, including ensembles with singular states,
\begin{equation}
 \exp\qty[-\frac{\alpha-1}{\alpha}I_\alpha^{{\rm Aug},\mbM}(P;W)]
 =\inf_{Y_x\geq0, \forall x\in\mX} \qty(\prod_x\qty(\Tr[\rho_xY_x^{-\frac{\alpha-1}{\alpha}}])^{P(x)})
 \left\|\sum_xP(x)Y_x\right\|^{\frac{\alpha-1}{\alpha}}.
 \label{eq:score-duality}
\end{equation}
\end{thm} 
\end{thmboxed}

\begin{proof}
In the following, let us denote $s\coloneqq (\alpha-1)/\alpha$.
Letters with $P(x)=0$ do not affect either side, so we consider $x\in\mX$ with $P(x)>0$. We first suppose that every $\rho_x$ is strictly positive and denote the right-hand side of Eq.~\eqref{eq:score-duality} by $F$.

Common rescaling $Y_x\mapsto tY_x$ leaves $F$ invariant.  If
$r=\norm{\sum_xP(x)Y_x}$, rescaling by $r^{-1}$ makes the norm equal to one, and the resulting product of trace factors equals the original scale-invariant objective.  Conversely, for every family satisfying
$\sum_xP(x)Y_x\leq I$, the original objective is no larger than its product of trace factors; if the norm is strictly smaller than one, scaling up to norm one makes these two quantities equal.  These two observations give
\begin{equation}
 F=\inf_{\substack{Y_x\geq 0\\\sum_xP(x)Y_x\leq I}}\prod_x\qty(\Tr[\rho_xY_x^{-s}])^{P(x)}.
 \label{eq:F-normalized}
\end{equation}
For positive scalars $a_x$ and $c_1,\ldots, c_{\aX}>0$ satisfying $\prod_xc_x^{P(x)}=1$, weighted AM--GM gives
\bal
\sum_xP(x)c_xa_x\geq \prod_xc_x^{P(x)}\prod_xa_x^{P(x)}=\prod_xa_x^{P(x)}.
\eal
The equality is attained if $c_xa_x=\prod_xa_x^{P(x)}$ holds.
From the observation above, we have the exact identity
\begin{equation}
 \prod_xa_x^{P(x)}
 =
 \min_{\substack{c_x>0\\\prod_xc_x^{P(x)}=1}}
 \sum_xP(x)c_xa_x
 \label{eq:weighted-amgm}
\end{equation}
with equality when $c_x=(\prod_ya_y^{P(y)})/a_x$ holds. From this, the optimization of $F$ is reduced to 
\bal
F&=\inf_{\substack{Y_x\geq 0\\\sum_xP(x)Y_x\leq I}}\prod_x\qty(\Tr[\rho_xY_x^{-s}])^{P(x)}\\
&=\inf_{\substack{Y_x\geq 0\\\sum_xP(x)Y_x\leq I}} \min_{\substack{c_x>0\\\prod_xc_x^{P(x)}=1}}\sum_xP(x)c_x\Tr[\rho_xY_x^{-s}]\\
&=\min_{\substack{c_x>0\\\prod_xc_x^{P(x)}=1}}\inf_{\substack{Y_x\geq 0\\\sum_xP(x)Y_x\leq I}}\sum_xP(x)c_x\Tr[\rho_xY_x^{-s}].
\eal

Fix $\Vec{c}\coloneqq (c_1,\ldots, c_\aX)$, and we focus the following optimization:
\bal
 \phi(\Vec{c})&\coloneqq \inf_{\substack{Y_x>0\\\sum_xP(x)Y_x\leq I}}\sum_xP(x)c_x\Tr[\rho_xY_x^{-s}]
\eal
Since $t\mapsto t^{-s}$ is operator convex for $0<s\le1$~\cite[Chapter~V]{bhatia_matrix_1997}, the optimization over the $Y_x$ is convex.  The constraint implies
$0\leq Y_x\leq P(x)^{-1}I$.  Because $\rho_x>0$, the objective diverges when an eigenvalue of some $Y_x$ tends to zero; consequently, the infimum is attained in the interior of a compact subset.  The strictly feasible choice $Y_x=\eta I$, $0<\eta<1$, verifies Slater's condition.  Strong finite-dimensional Lagrange duality therefore gives, with $Z\geq 0$,
\begin{align}
 \phi(\Vec{c})
 &\coloneqq \inf_{\substack{Y_x>0\\\sum_xQ(x)Y_x\leq I}}
 \sum_xQ(x)c_x\Tr[\rho_xY_x^{-s}]
 \notag\\
 &=\sup_{Z\geq 0}
 \left\{-\Tr Z+\sum_xP(x)
 \inf_{Y>0}
 \bigl(c_x\Tr[\rho_xY^{-s}]+\Tr[ZY]\bigr)\right\}.
 \label{eq:lagrange}
\end{align}
Now, let us rescale $Y$ as $Y= tY'$. Since taking the infimum over $Y>0$ is equivalent to taking the infimum over $t>0$ and $Y'>0$, we have
\bal
\inf_{Y>0}\qty(c_x\Tr[\rho_xY^{-s}]+\Tr[ZY])&=\inf_{Y'>0}\inf_{t>0}\qty(c_xt^{-s}\Tr[\rho_xY'^{-s}]+t\Tr[ZY']).\\
\eal
Since the infimum over $t$ for a fixed $Y'$ is attained with $t^*=(sc\Tr[\rho_xY'^{-s}]/\Tr[ZY'])^{1/(s+1)}$, substituting it yields
\bal
\inf_{Y'>0}\inf_{t>0}\qty(c_xt^{-s}\Tr[\rho_xY'^{-s}]+t\Tr[ZY'])&=\inf_{Y'>0}\alpha^{-\alpha}(1-\alpha)^{-(1-\alpha)}c_x^\alpha\Tr[\rho_xY'^{-s}]^\alpha\Tr[ZY']^{(1-\alpha)}.
\eal
Due to the variational form of the measured quasi-entropy in Eq.~\eqref{eq:variational measured renyi}, we have
\bal
\inf_{Y'>0}\alpha^{-\alpha}(1-\alpha)^{-(1-\alpha)}c_x^\alpha\Tr[\rho_xY'^{-s}]^\alpha\Tr[ZY']^{(1-\alpha)}=C_\alpha c_x^\alpha Q_\alpha^{\mbM}(\rho_x\Vert Z), \qquad C_\alpha\coloneqq \alpha^{-\alpha}(1-\alpha)^{-(1-\alpha)}.
\eal

Writing $Z=t\sigma$ with $\sigma\in\mD(\mH)$ and using
$Q_\alpha^{\mbM}(\rho_x\Vert t\sigma)=t^{1-\alpha}Q_\alpha^{\mbM}(\rho_x\Vert\sigma)$, we have
\bal
  \phi(\Vec{c})&=\sup_{Z\geq 0}\qty{-\Tr Z+\sum_xP(x)C_\alpha c_x^\alpha Q_\alpha^{\mbM}(\rho_x\Vert Z)}\\
  &=\sup_{\sigma\in\mD(\mH)}\sup_{t\geq 0}\qty{-t+t^{1-\alpha}\sum_xP(x)C_\alpha c_x^\alpha Q_\alpha^{\mbM}(\rho_x\Vert \sigma)}.
\eal
If we explicitly calculate the supremum over $t$, the constant $C_\alpha$ cancels out, and we have
\begin{equation}
 \phi(\Vec{c})
 =\max_{\sigma\in\mD(\mH)}
 \left(
   \sum_xP(x)c_x^\alpha Q_\alpha^{\mbM}(\rho_x\Vert\sigma)
 \right)^{1/\alpha}.
 \label{eq:phi}
\end{equation}

Set $c_x=e^{v_x}$ with $v_x\in\mbR$. Then, the constraint on $c_x$ yields $\sum_xP(x)v_x=0$. The function
\bal
 f(v,\sigma)\coloneqq \sum_xP(x)e^{\alpha v_x}
 Q_\alpha^{\mbM}(\rho_x\Vert\sigma)
\eal
is continuous and convex in $v$.  For a fixed $v$, it is concave and upper semicontinuous in $\sigma$: by Eq.~\eqref{eq:variational measured renyi}, each measured quasi-entropy is an infimum of affine continuous functions of $\sigma$.  The state space is compact and convex, while the affine hyperplane for $v$ is convex.  Sion's minimax theorem~\cite{sion_1958_general_minimax_theorems} therefore applies and gives
\begin{align}
 F^\alpha
 &=\inf_{\substack{v\in\mathbb R^{\mX}\\\sum_xP(x)v_x=0}}\max_{\sigma\in\mD(\mH)} f(v,\sigma)
 \notag\\
 &=\max_{\sigma\in\mD(\mH)}
 \inf_{\substack{v\in\mathbb R^{\mX}\\\sum_xP(x)v_x=0}}
 f(v,\sigma)
 \notag\\
 &=\max_{\sigma\in\mD(\mH)}
 \prod_xQ_\alpha^{\mbM}(\rho_x\Vert\sigma)^{P(x)}.
 \label{eq:minimax}
\end{align}
Here, the last equality is due to Eq.~\eqref{eq:weighted-amgm} and $c_x=e^{\alpha v_x}$.  Since
$Q_\alpha^{\mbM}(\rho\Vert\sigma)=
\exp[(\alpha-1)D_\alpha^{\mbM}(\rho\Vert\sigma)]$ and
$\alpha-1<0$, it holds that
\bal \label{eq:quasi-radius}
\max_{\sigma\in\mD(\mH)}\prod_xQ_\alpha^{\mbM}(\rho_x\Vert\sigma)^{P(x)} &=\exp\qty[(\alpha-1)\min_{\sigma\in\mD(\mH)}\sum_{x\in\mX}P(x)D^\mbM(\rho_x\|\sigma)]\\
 &=\exp\qty[(\alpha-1)I^{{\rm Aug}, \mbM}_\alpha(P;W)].
\eal
Taking the $1/\alpha$ power proves Eq.~\eqref{eq:score-duality} for full-rank states.

Let us discuss the general scenario where $\rho_x$ is not necessarily full-rank.
For arbitrary states, put $\rho_x^{(\varepsilon)}=(1-\ve)\rho_x+\ve I/d$.
It suffices to show that the limits as $ \varepsilon\to 0$ on both sides coincide with the quantities defined without the perturbation.
Let $F_\varepsilon$ be the right-hand side of Eq.~\eqref{eq:score-duality} with the states $\qty{\rho_x^{(\varepsilon)}}_x$. Then, it holds that
\bal
 \Tr[\rho_x^{(\varepsilon)}Y_x^{-s}]
 \ge(1-\varepsilon)\Tr[\rho_xY_x^{-s}],
\eal
so $F_\varepsilon\ge(1-\varepsilon)F$, which means that we have $\liminf_{\ve\to 0}F_\ve\geq F$. 
On the other hand, fix $\forall\delta>0$. Then, there exists a family $\qty{Y^\delta_x}_{x\in\mX}$ of operators such that 
\bal
 \qty(\prod_x\qty(\Tr[\rho_x(Y^\delta_x)^{-\frac{\alpha-1}{\alpha}}])^{P(x)})
 \left\|\sum_xP(x)Y_x^\delta\right\|^{\frac{\alpha-1}{\alpha}}\leq F+\delta.
\eal
Then, it holds that
\bal
\Tr[\rho_x^{(\varepsilon)}(Y_x^{\delta})^{-\frac{\alpha-1}{\alpha}}]
&=(1-\varepsilon)
\Tr[\rho_x(Y_x^{\delta})^{-\frac{\alpha-1}{\alpha}}]
+\frac{\varepsilon}{d}\Tr[(Y_x^{\delta})^{-\frac{\alpha-1}{\alpha}}]\\
&\xrightarrow{\ve\to +0} \Tr[\rho_x(Y_x^{\delta})^{-\frac{\alpha-1}{\alpha}}], 
\eal
implying that
\bal
\limsup_{\ve\to 0}\qty(\prod_x\qty(\Tr[\rho_x^{(\varepsilon)}(Y^\delta_x)^{-\frac{\alpha-1}{\alpha}}])^{P(x)})
 \left\|\sum_xP(x)Y_x^\delta\right\|^{\frac{\alpha-1}{\alpha}}=\qty(\prod_x\qty(\Tr[\rho_x(Y^\delta_x)^{-\frac{\alpha-1}{\alpha}}])^{P(x)})
 \left\|\sum_xP(x)Y_x^\delta\right\|^{\frac{\alpha-1}{\alpha}}.
\eal
Since $F_\ve$ is defined as the infimimum over $\qty{Y_x}_{x\in\mX}$, we have $\limsup_{\ve\to 0}\leq F_0$. Combining these, we have $\lim_{\ve \to +0}F_\ve=F_0$.

For the quasi-radius in Eq.~\eqref{eq:quasi-radius}, the variational formula
Eq.~\eqref{eq:variational measured renyi} shows that
$Q_\alpha^{\mbM}$ is jointly upper semicontinuous and monotone in its first argument in Loewner order.  Therefore the maximum in
Eq.~\eqref{eq:quasi-radius} is attained and its value, denoted $R_\varepsilon$, satisfies
$R_\varepsilon\ge(1-\varepsilon)^\alpha R_0$.  If
$\sigma_\varepsilon$ is a maximizer, compactness gives a convergent subsequence
$\sigma_{\varepsilon_j}\to\sigma_\natural$; joint upper semicontinuity then gives
\bal
 \limsup_{j\to\infty}R_{\varepsilon_j}
 \le
 \prod_uQ_\alpha^{\mbM}(\rho_u\Vert\sigma_\natural)^{Q(u)}
 \le R_0.
\eal
Thus $R_\varepsilon\to R_0$.  Passing to the limit in the full-rank identity completes the proof.
\end{proof}

The family $\qty{Y_x}_{x\in\mX}$ of positive semidefinite matrices appearing in the optimization of Theorem~\ref{thm: measured-augustin duality} turns out to be a potential candidate for the family of operators that forms a decoder in Proposition~\ref{pro: quotient error bound} and gives us the bound with respect to the measured Rényi Augustin information.

\begin{proboxed}
    \begin{pro}[Upper bound by the measured Rényi Augustin information]
\label{pro: product-score}
Fix $\alpha\in[1/2,1)$ and $s=(1-\alpha)/\alpha$.  Suppose that a family of positive semidefinite matrices $\qty{Y_x}_{x\in\mX}$ satisfies
\begin{equation}
   \qty(\prod_x\qty(\Tr[\rho_xY_x^{-s}])^{P(x)})\left\|\sum_xP(x)Y_x\right\|^s\leq
 e^\ve
 \exp\bigl[-sI_\alpha^{{\rm Aug},\mbM}(P;W)\bigr].
 \label{eq:near-score}
\end{equation}
For a string $x^n\in T^n_P$ with type $P$, we define
\bal
 T_{\mathbf Y}(x^n)\coloneqq \bigotimes_{i=1}^NY_{x_i},
 \qquad
 \overline T_{\mathbf Y}\coloneqq 
 \frac1{\abs{T_P^n}}
 \sum_{x^n\in T_P^n}T_{\mathbf Y}(x^n),\qquad\widehat T_{\mathbf Y}(x^n)\coloneqq 
 \frac{T_{\mathbf Y}(x^n)}{\norm{\overline T_{\mathbf Y}}}.
\eal
Then, with $p_{n,P}\coloneqq P^{\otimes n}(T_P^n)$,
\begin{equation}
 \Tr[\rho_{x^n}\widehat T_{\mathbf Y}(x^n)^{-s}]
 \le
 p_{n,P}^{-s}e^{n\ve}
 \exp\bigl[-nsI_\alpha^{{\rm Aug},\mbM}(P;W)\bigr]
 \label{eq:product-score-bound}
\end{equation}
for every $x^n\in T_P^n$.
\end{pro}
\end{proboxed}

\begin{proof}
The type constraint gives
\bal
\Tr[\rho_{x^n}T_{\mathbf Y}(x^n)^{-s}]
 =\prod_x\bigl(\Tr[\rho_xY_x^{-s}]\bigr)^{nP(x)}.
\eal
On the other hand, expanding the tensor product of $(\sum_xP(x)Y_x)^{\otimes n}$ gives
\bal
\left(\sum_xP(x)Y_x\right)^{\otimes n}\geq  p_{n,P}\overline T_{\mathbf Y},
\eal
which yields
\bal
\left\|\overline T_{\mathbf Y}\right\|
 \le p_{N,Q}^{-1}
 \left\|\sum_xP(x)Y_x\right\|^n.
\eal

Multiplying these two and using Eq.~\eqref{eq:near-score} proves the claim.
\end{proof}

\subsubsection{The universal family of positive operators}
Even though one can upper-bound the decoding error probability by constructing the quotient decoder from the family $\qty{\widehat{T}_{\bf Y}(x^n)}_{x^n\in T^n_P}$ of positive definite operators, the choice of the family $\qty{Y_x}_{x\in\mX}$ generally depends on the details of the given classical-quantum channel $W$, which can be observed from Eq.~\eqref{eq:score-duality}. In the following, we show that there exists a family of operators that universally upper-bounds $\qty{\widehat{T}_{\bf Y}(x^n)}_{x^n\in T^n_P}$ up to a polynomial factor, using Schur-Weyl duality, a representation-theoretic tool.

For a blocklength $n$ and a fixed $n$-type $P$, we take a canonical string $x^n_P$ as
\bal
x^n_P\coloneqq 1^{m_1}2^{m_2}\cdots \abs{\mX}^{m_\abs{\mX}},
\eal
where $m_x\coloneqq nP(x)$. In the following, we only consider the symbol $x\in\mX$ with $P(x)>0$.
Let $H_P$ denote the subgroup of the symmetric group $\mfS_n$ that stabilizes the canonical string $x^n_P$. $H_P$ is written as
\bal
H_P&\coloneqq\qty{\pi\in\mfS_n~|~\pi(x^n_P)=x^n_P}\\
&=\prod_{x\in\mX}\mfS_{m_x}.
\eal

For each $x$, Schur--Weyl duality gives
\begin{equation}
 \mH^{\otimes m_x}=\bigoplus_{\lambda_x\in Y_d^{m_x}}\mW_{\lambda_x}\otimes\mU_{\lambda_x},
 \label{eq:schur-block}
\end{equation}
where $\mW_{\lambda_x}$ is irreducible representation of the $GL(d)$, and $\mU_{\lambda_x}$ is the irreducible representation of the symmetric group $\mfS_{m_x}$. 
Consequently, the $H_P$-commutant $\mA_P$, the algebra of operators commutative with the action of $H_Q$, is decomposed as 
\begin{equation}
 \mA_P=\bigoplus_{\Vec{\lambda}}\mB\qty(\mW_{\vlam})\otimes I_{\mU_{\vlam}},
 \label{eq:commutant}
\end{equation}
where $\vlam\coloneqq (\lambda_{m_1},\ldots, \lambda_{m_\aX})$ is the tuple of the young tableau, and $\mW_{\vlam}$ and $\mU_{\vlam}$ are the tensor products of each irreps.~defined as
\bal
\mW_{\vlam}\coloneqq \bigotimes_{x\in\mX}\mW_{\lambda_x}, \qquad   \mU_{\vlam}\coloneqq \bigotimes_{x\in\mX}\mU_{\lambda_x}.
\eal
An element $A$ of $\mA_P$ has the reduced block form
\bal
A=\bigoplus_{\vlam}A_\vlam\otimes I_{\mU\vlam}.
\eal

The following shows that there exists a universal family $\qty{\Sigma_{x^n}}_{x^n\in T_P^n}$ of operators which can dominate any family $\qty{\widehat T_{\mathbf Y}(x^n)}_{x^n}$ in Proposition~\ref{pro: product-score} up to a polynomial factor, and thus gives us the universal decoder.
\begin{thmboxed}
    \begin{thm}[Universal family of operators for the universal decoder]
\label{thm:schur-universal}
There is a family $\qty{\Sigma_{x^n}}_{x^n\in T_P^n}$ of positive definite operators depending only on $n,P,d$, satisfying the permutation covariance $U_{\pi}\Sigma_{x^n}U^\dagger_{\pi}=\Sigma_{\pi(x^n)}$, the normalization condition,
\begin{equation}
\left\|\overline\Sigma\right \|\leq1,\qquad \overline\Sigma\coloneqq \frac1{\abs{T_P^n}}
 \sum_{x^n\in T_P^n}\Sigma_{x^n},
 \label{eq:sigma-mean}
\end{equation}
and, for every family of positive operators $\mathbf Y=\qty{Y_x}_{x\in\mX}$,
\begin{equation}
 \widehat T_{\mathbf Y}(x^n)\leq  (n+1)^{\aX(d+2)(d-1)/2}\Sigma_{x^n}, \qquad \forall x^n\in T^n_P.
 \label{eq:universal-domination}
\end{equation}
\end{thm}
\end{thmboxed}

\begin{proof}
Let $L\coloneqq \abs{T_P^n}$.  The $H_P$-commutant $\mA_P$ is a finite-dimensional real vector space when restricted to its Hermitian part. Let us consider a set of operators
\begin{equation}
 \mathscr K\coloneqq \overline{\left\{
 \widehat T_{\mathbf Y}(x^n_P):Y_x>0
 \right\}}\subset\mA_P.
 \label{eq:K-set}
\end{equation}
Here, note that it is compact. Indeed, Since it holds that $\overline T_{\mathbf Y}\geq \frac1L T_{\mathbf Y}(x^n_Q),$ we have $ 0\leq \widehat T_{\mathbf Y}(x^n_Q)\leq LI.$
Thus the defining set is bounded, and its closure is compact in finite dimension.  Moreover, $I\in\mathscr K$ by choosing every $Y_x=I$.

Let $\mathscr C\coloneqq \operatorname{conv}\mathscr K$.  It is again compact in finite dimension. Take an element of $S\in\mathscr{C}$, and then $S$ is decomposed as 
\bal
 S=\bigoplus_{\vlam}S_{\vlam}\otimes I_{\mU_{\vlam}}
\eal
On the strictly positive part of $\mathscr C$, define the reduced log-determinant
\begin{equation}
 \Phi(S)\coloneqq \sum_{\boldsymbol\lambda}
 \log\det S_{\boldsymbol\lambda},
 \label{eq:reduced-logdet}
\end{equation}
 Also, we set $\Phi(S)=-\infty$ for any singular matrix $S$.  Note that $\Phi(S)$ is continuous as an extended-real-valued function. Since $I\in\mathscr C$, compactness gives a maximizer $S_\star\in\mathscr C$ with every reduced block strictly positive.

Fix an arbitrary $T\in\mathscr K$, and consider the segment $S_t=(1-t)S_\star+tT$.
Noting that the Fr\'echet deriverate of the function $f(X)\coloneqq \log \det X$ is given by $Df(X)[H]=\Tr[X^{-1}H]$~\cite{bhatia_matrix_1997},
the one-sided directional derivative at $t=0$ is nonpositive and equals
\bal
 0\ge
 \left.\frac{d}{dt}\Phi(S_t)\right|_{t=0+}
 &= \sum_{\vlam}
 \Tr[S_{\star,\vlam}^{-1}(T_{\vlam}-S_{\star,\vlam})]\\
 &=\sum_{\vlam}
 \Tr[S_{\star,\vlam}^{-1}T_{\vlam}]-\sum_{\vlam}\dim\mW_{\vlam},\\
\sum_{\vlam}\Tr[S_{\star,\vlam}^{-1}T_{\vlam}]&\leq \sum_{\vlam}\dim\mW_{\vlam}=:\kappa_{n,  d}.
\eal
From this, focusing on each block associated with $\vlam$, we have
\bal
 \lambda_{\max}(S_{\star,\vlam}^{-1/2}T_{\vlam}S_{\star,\vlam}^{-1/2})\leq \Tr[S_{\star,\boldsymbol\lambda}^{-1}T_{\vlam}]
 \le\kappa_{n,d},
\eal
implying that 
\begin{equation}
 T\leq \kappa_{n,  d}S_\star\qquad \forall T\in\mathscr K.
 \label{eq:logdet-domination}
\end{equation}

Let us define $\Sigma_{x^n_P}\coloneqq S_\star$. For other strings $x^n\in T^n_P$, we define $\Sigma_{x^n}$ as
\bal
\Sigma_{x^n}\coloneqq U_\pi\Sigma_{x^n_P}U^\dagger_\pi
\eal
with the permutation $\pi\in\mfS_n$ satisfying $\pi(x^n_P)=x^n$. 
Here, note that the definition of $\Sigma_{x^n}$ is well-defined because $\Sigma_{x^n_P}$ is stabilized by the action of $H_P$.
Moreover, the permutation covariance of the family $\qty{\Sigma_{x^n}}_{x^n\in T^n_P}$ of operators follows directly from the definition.

It remains to verify the orbit-mean normalization.  Let us define the twirling map
\bal
\mE_P(T)\coloneqq \frac{1}{\abs{T^n_P}}\sum_{[\pi]\in\mfS_n/H_P}U_\pi TU^\dagger_\pi
\eal
Again, this linear map is well-defined and continuous on $\mA_P$. Let us take an operator $\widehat T_{\mathbf Y}(x^n_P)$, and then it holds that 
\bal
 \mathcal E_P\!\left(\widehat T_{\mathbf Y}(x^n_P)\right)
 =\frac{\overline T_{\mathbf Y}}{\norm{\overline T_{\mathbf Y}}},
\eal
whose norm is one.  Continuity gives
$\norm{\mathcal E_P(T)}\le1$ for every $T\in\mathscr K$, and convexity of the norm gives the same bound for every $T\in\mathscr C$.  Since $S_\star\in\mathscr C$,
\bal
 \norm{\overline\Sigma}
 =\norm{\mathcal E_P(S_\star)}\le1,
\eal
which is Eq.~\eqref{eq:sigma-mean}.  Conjugating
Eq.~\eqref{eq:logdet-domination} proves Eq.~\eqref{eq:universal-domination}.

Finally, let us obtain the upper bound on $\kappa_{n, d}\coloneqq \sum_{\vlam}\dim\mW_\vlam$.
Due to Weyl's character formula and the upper bound on the possible types in Eq.~\eqref{eq: type counting}, we have
\bal
 \abs{Y^{m_x}_d}\le(m_x+1)^{d-1},
 \qquad
 \dim\mW_{\lambda_x}\le(m_x+1)^{d(d-1)/2}.
\eal
Therefore
\bal
 \sum_{\lambda_x\in Y_d^{m_x}}\dim\mW_{\lambda_x}
 \le(m_x+1)^{(d+2)(d-1)/2}.
\eal
Taking the product over $x\in\mX$, we have
\bal
\kappa_{n, d}&\coloneqq\sum_{\vlam}\dim\mW_\vlam\\
&=\prod_{x\in\mX}\sum_{\lambda_x\in Y_d^{m_x}}\dim\mW_{\lambda_x}\\
&\leq \prod_{x\in\mX}(m_x+1)^{(d+2)(d-1)/2}\leq (n+1)^{\aX(d+2)(d-1)/2},
\eal
which concludes the proof.
\end{proof}

\subsection{Regularization to the sandwiched random coding bound}
\subsubsection{Regularization from the measured to the sandwiched Rényi Augustin information}
In the discussions so far, we see that the random-coding bound with the measured Rényi Augustin information is universally achievable. In the following, we apply the regularization to lift the measured Rényi Augustin information to the sandwiched Rényi information, which yields the claim of Theorem~\ref{supthm: universal random coding bound}.
\begin{proboxed}
\begin{pro}[Finite-block regularization, including $\alpha=1/2$]
\label{prop:regularization}
Let $W:\mX\to\mD(\mH)$ be a classical-quantum channel and $P$ be a probability distribution on $\mX$.
Then, it holds that
\bal\label{eq: regularized measeured augustin}
\lim_{\ell\to\infty}\frac1\ell
 I_\alpha^{{\rm Aug},\mbM}(P^{\otimes\ell};W^{\otimes\ell})
 =\widetilde I_\alpha^{{\rm Aug}}(P;W).
\eal
 for every $\alpha\in[1/2,1)$.
\end{pro}
\end{proboxed}

\begin{proof}
We first show that
\begin{equation}
 \ell\widetilde I_\alpha^{{\rm Aug}}(P;W)-\log v_{\ell,d}\leq
 I_\alpha^{{\rm Aug},\mbM}(P^{\otimes\ell};W^{\otimes\ell})\leq\ell\widetilde I_\alpha^{{\rm Aug}}(P;W)
 \label{eq: regularization-interior}
\end{equation}
holds for $\alpha\in(1/2,1)$. Noting that the measured divergence is no larger than the sandwiched divergence, evaluating the measured Augustin objective at a product sandwiched-Augustin center gives the upper bound in Eq.~\eqref{eq: regularization-interior}.

Let us show the lower bound. Choose a measured-Augustin minimizer $\omega_\ell$.  Such a minimizer exists because the objective is lower semicontinuous on the compact state space and is finite at every full-rank state.  The objective is convex in its second argument: $Q_\alpha^{\mbM}(\rho\Vert\sigma)$ is concave in $\sigma$ by Eq.~\eqref{eq:variational measured renyi}, while $t\mapsto(\alpha-1)^{-1}\log t$ is convex and decreasing. 
Thus, without loss of generality one can take $\omega_\ell$ as a permutation-symmetric state.
Schur--Weyl duality shows that such an operator has at most
\bal
 \sum_{\lambda\in Y_d^\ell}\dim\mW_\lambda
 \le (\ell+1)^{(d+2)(d-1)/2}
\eal
distinct eigenvalues.  Let $\mathcal P_{\omega_\ell}$ be its spectral pinching and let $v$ be the number of spectral projections.  Measuring in a common eigenbasis of $\mathcal P_{\omega_\ell}(\rho)$ and $\omega_\ell$ gives
\bal
 D_\alpha^{\mbM}(\rho\Vert\omega_\ell)
 \ge
 \widetilde D_\alpha(\mathcal P_{\omega_\ell}(\rho)\Vert\omega_\ell).
\eal
Here, due to~\cite[Lemma 3]{Hayashi_2016_correlation_detection}, we have
\bal
\widetilde D_\alpha(\mathcal P_{\omega_\ell}(\rho)\Vert\omega_\ell)
 \geq
 \widetilde D_\alpha(\rho\Vert\omega_\ell)-\log \abs{\spec(\omega_\ell)}.
\eal
Averaging over $P^{\otimes\ell}$ and taking the minimization over the state in the second argument of the right-hand side, we have
\bal
 I_\alpha^{{\rm Aug},\mbM}(P^{\otimes\ell};W^{\otimes\ell})
 \ge
 \widetilde I_\alpha^{{\rm Aug}}(P^{\otimes\ell};W^{\otimes\ell})-
 \log v_{\ell,d}.
\eal
Since the weighted sandwiched radius is additive for $\alpha>1/2$~\cite[Corollary~III.22]{Mosonyi_2021_divergence_radii}, we reach Eq.~\eqref{eq: regularization-interior}.
Dividing by $\ell$, and taking the limit $\ell\to\infty$, we reach the claim for $\alpha\in(1/2, 1)$.

At $\alpha=1/2$, the Fuchs--Caves fidelity theorem~\cite{fuchs_mathematical_1995} gives
\bal
 D_{1/2}^{\mbM}(\rho\Vert\sigma)
 =\widetilde D_{1/2}(\rho\Vert\sigma)
\eal
for every pair of states.  Thus the measured and sandwiched Augustin objectives coincide exactly on every block.  Additivity at $\alpha=1/2$ follows by taking the left-hand limit $\beta\to 1/2$ in the identity of \cite[Corollary~III.22]{Mosonyi_2021_divergence_radii} and using finite-dimensional continuity of sandwiched Augustin information in the order~\cite[Proposition~5 and Theorem~13]{Cheng_properties_2022}. From this, we can see that Eq.~\eqref{eq: regularized measeured augustin} holds for $\alpha=1/2$.
\end{proof}

\subsubsection{Combining everything altogether: Proof of Theorem~\ref{supthm: universal random coding bound} through double-blocking}
Finally, we prove Theorem~\ref{supthm: universal random coding bound}. We first fix the parameters as follows:
For a sequence $\qty{P_n}_n$ of $n$-types $P_n$, let $h_n$ be the least common multiple of the reduced denominators of the numbers $P_n(1), \ldots, P_n(\aX)$.  For admissible final blocklengths $n$, choose
\begin{equation}
 \ell_n\coloneqq \max\left\{1,\left\lfloor(\log(n+e))^{1/3}\right\rfloor\right\},
 \qquad
 N_n\coloneqq h_n^{\ell_n}
 \left\lfloor\frac{n}{\ell_nh_n^{\ell_n}}\right\rfloor,
 \qquad
 m_n\coloneqq N_n\ell_n.
 \label{eq:block-schedule}
\end{equation}
For any blocklength $n$, we see $\ell_n$ symbols as a single supersymbol in the set of the alphabet $\mU_n\coloneqq \mX^{\ell_n}$, and the blocklength for the new alphabet is written as $N_n$. We ignore the remaining $(n-m_n)$ symbols.
Using the first $m_n$ symbols, we take a codebook with the constant composition $Q_n\coloneqq P_n^{\otimes\ell_n}$.
We also denote $D_n\coloneqq d^{\ell_n}$.
The outer target rate is $\sR\coloneqq \ell_nR<H(Q_n)=\ell_nH(P_n)$.

Due to Lemma~\ref{lem: packing}, we can take $(Q_n, \sR, \widetilde{\delta}_1'(n))$-good codebook $ \mB_n\subseteq T_{Q_n}^{N_n}$ with $\widetilde{\delta}'(n)=\Omega\qty((\log n)^{\aU^2})$, meaning that one can take $\widetilde{\delta}'(n)$ so that it scales as $\widetilde{\delta}'(n)=o(n)$.

Now, due to \Cref{thm:schur-universal}, one can take the family $\qty{\Sigma_{u^{N_n}}}_{u^{N_n}\in T^{N_n}_{Q_n}}$ of positive operators on the Hilbert space $\mH^{m_n}$ such that for any family $\qty{Y_u}_{u\in\mU}$ of positive definite operators on $\mH^{\ell_n},$
\begin{equation}\label{eq: block universal operator}
 \widehat T_{\mathbf Y}(u^{N_n})\leq  (N_n+1)^{\abs{\mU_n}(D_n+2)(D_n-1)/2}\Sigma_{u^{N_n}}, \qquad \forall u^{N_n}\in T^{N_n}_{Q_n}
\end{equation}
Note that the overhead $(N_n+1)^{\abs{\mU_n}(D_n+2)(D_n-1)/2}$ is subexponential $o(e^n)$ in $n$.
For the actual blocklength $n$, we take the decoder $\qty{\Lambda_m}_{m}$ from the universal family $\qty{\Sigma_{u^{N_n}}}_{u^{N_n}\in T^{N_n}_{Q_n}}$ as 
\bal
\Lambda_m\coloneqq \frac{\Sigma_{u^{N_n}(m)}}{\sum_{m'\in\mM_n}\Sigma_{u^{N_n}(m')}}\otimes I^{\otimes (n-m_n)}.
\eal
Also, note that 
\bal
Q_n(T^{N_n}_{Q_n})\geq \frac{1}{(N_n+1)^{\abs{\mU_n}}}=\frac{1}{o(e^n)}.
\eal
holds~\cite{Csiszar_Korner_2011}.
Now, we are finally ready to prove the main statement.

\begin{proof}[Proof of \Cref{supthm: universal random coding bound}]
For $R=0$, use a single fixed type-$P$ codeword; its error is zero.  Assume $R>0$, and fix an arbitrary cq channel $W$ and an arbitrary $\alpha\in[1/2,1)$.  Let us denote $s=(1-\alpha)/\alpha$.
By \Cref{thm: measured-augustin duality}, one can take a family $\qty{Y_u}_u$ of positive operators satisfying Eq.~\eqref{eq:near-score} with $\varepsilon_n\coloneqq 1/N_n$. Since  $x\mapsto x^{-s}$ is operator anti-monotone for $0<s\le1$, Eq.~\eqref{eq: block universal operator} yields
\begin{equation}
\Sigma_{u^{N_n}}^{-s}\leq \qty((N_n+1)^{\abs{\mU_n}(D_n+2)(D_n-1)/2})^s  \widehat T_{\mathbf Y}(u^{N_n})
 \label{eq:inverse-domination}
\end{equation}
Denoting $\kappa_n\coloneqq (N_n+1)^{\abs{\mU_n}(D_n+2)(D_n-1)/2}$, note that $\kappa_n=o(e^n)$.

Applying Proposition~\ref{pro: quotient error bound}, we have
\begin{align}
 1-\Tr[\rho_{x^n(m)}\Lambda_m]
 &\le
 (M_n-1)^se^{s\widetilde{\delta}'(n)}
 \kappa_n^s
 \Tr[\rho_{u^{N_n}(m)}\widehat  T_{\mathbf Y}(u^{N_n}(m))^{-s}]\\
 &\leq (M_n-1)^se^{s\widetilde{\delta}'(n)}
 \kappa_n^s~ Q_n(T^{N_n}_{Q_n})^{-s}\exp\left[-N_ns
 I_\alpha^{{\rm Aug},\mbM}
 (P^{\otimes\ell_n};W^{\otimes\ell_n})\right].
 \label{eq:master-bound}
\end{align}
Note that the term $e^{s\widetilde{\delta}'(n)}
\kappa_n^s ~Q_n(T^{N_n}_{Q_n})^{-s}$ does not contribute to the exponent because it is at most subexponential.
 On the other hand, since we take $M_n=2^{nR}$, it holds that 
 \bal
  s\log(M_n-1)\le snR.
 \eal
Combining these, we have
\bal
\liminf_{n\to\infty}-\frac{1}{n}\log P^n_e(P_n, R, W)\geq s\frac{m_n}{n}\qty(\frac{1}{\ell_n}I_\alpha^{{\rm Aug},\mbM}
 (P_n^{\otimes\ell_n};W^{\otimes\ell_n})-\frac{n}{m_n}R)
\eal
for any $n\in\mN$. 
Due to the limit $P_n\to P$, $\ell_n\to\infty$, $\frac{m_n}{n}\to 1$ in the asymptotic limit $n\to\infty$ and proposition~\ref{prop:regularization}, we have
\bal
\liminf_{n\to\infty}-\frac{1}{n}\log P^n_e(P_n, R, W)\geq s\bigl(\widetilde I_\alpha^{{\rm Aug}}(P;W)-R\bigr).
\eal
Taking the supremum over $\alpha\in[1/2,1)$, we reach the claim.

\end{proof}

\end{document}